\documentclass{article}
\usepackage[margin=1in]{geometry}
\usepackage[utf8]{inputenc}

\usepackage{tikz}
\usetikzlibrary{snakes,arrows,shapes,calc,positioning}
\usepackage{amsmath, amssymb, graphicx, listings, verbatim, enumerate, relsize, algorithmicx, tikz}
\usepackage{enumitem}
\usepackage{hyperref}
\usepackage{cleveref}
\usepackage{mathdots}
\usepackage{mathrsfs}
\usepackage{mathtools}
\usepackage{subcaption}
\usepackage{diagbox}
\usepackage{amsthm}
\usepackage{thmtools,thm-restate}
\usepackage{cite}

\usepackage[noend]{algpseudocode}
\usepackage[ruled, vlined]{algorithm2e}

\usepackage[textsize=tiny,textwidth=0.6in]{todonotes}

\usepackage{authblk}

\title{Bandwidth Optimal Pipeline Schedule for Collective Communication}

\author{Liangyu Zhao}
\author{Arvind Krishnamurthy}

\affil{University of Washington}
\date{}

\newcommand{\brac}[1]{\left[#1\right]}
\newcommand{\abs}[1]{\left|#1\right|}

\DeclareMathOperator{\R}{\mathbb{R}}
\DeclareMathOperator{\Q}{\mathbb{Q}}
\DeclareMathOperator{\Z}{\mathbb{Z}}
\DeclareMathOperator{\N}{\mathbb{N}}
\DeclareMathOperator{\I}{\mathbb{I}}
\DeclareMathOperator{\cF}{\mathcal{F}}
\DeclareMathOperator{\cV}{\mathcal{V}}

\DeclareMathOperator{\cO}{\mathcal{O}}
\DeclareMathOperator{\cE}{\mathcal{E}}

\DeclareMathOperator*{\argmax}{arg\,max}

\SetKwFunction{emap}{routing}

\setlist[enumerate]{itemsep=0pt}
\setlist[itemize]{itemsep=0pt}

\begin{document}

\maketitle

\begin{abstract}
    We present a strongly polynomial-time algorithm to generate bandwidth optimal allgather/reduce-scatter on any network topology, with or without switches. Our algorithm constructs pipeline schedules achieving provably the best possible bandwidth performance on a given topology. To provide a universal solution, we model the network topology as a directed graph with heterogeneous link capacities and switches directly as vertices in the graph representation. The algorithm is strongly polynomial-time with respect to the topology size. This work heavily relies on previous graph theory work on edge-disjoint spanning trees~\cite{bang-jensen, tarjan, edmonds, berczi2010packing, schrijver} and edge splitting~\cite{bang-jensen,frank,jackson}. While we focus on allgather, the methods in this paper can be easily extended to generate schedules for reduce, broadcast, reduce-scatter, and allreduce.
\end{abstract}

\section{Introduction}

In this section, we explain the foundational concepts necessary for introducing the schedule generation algorithm. In \S\ref{sec:aglowerbound}, we introduce a lower bound for the bandwidth runtime of allgather. This lower bound is also the optimal bandwidth runtime we aim to achieve. In \S\ref{sec:pipeline}, we briefly introduce the pipeline schedule. In \S\ref{sec:dilemma}, we describe a dilemma that makes pipeline schedule the only way to achieve bandwidth optimality.

\subsection{Allgather Lower Bound}\label{sec:aglowerbound}

Let $G=(V=V_s\cup V_c,E)$ be an arbitrary network topology, where $V_s$ are the \textit{switch nodes}, and $V_c$ are the \textit{compute nodes}. In allgather, only nodes in $V_c$ need to broadcast and receive data. Given any partition $S$ of $G$, every compute node in $S$ needs to send its shard of data out of $S$ if $S$ does not include all compute nodes. Thus, a lower bound for the allgather bandwidth runtime $T_B$ in $G$ is
\begin{equation}\label{eq:lowerbound}
    T_B\geq\frac{M}{N}\max_{S\subset V,S\not\supseteq V_c}\frac{|S\cap V_c|}{B^+_G(S)},
\end{equation}
where $M$ is the total data size, $N=|V_c|$, and $B^+_G(S)$ is the total bandwidth exiting the partition $S$. In this paper, we will show how to construct pipeline schedule that achieves this lower bound. Thus, as an ancillary result, lower bound (\ref{eq:lowerbound}) is exactly the optimal allgather bandwidth runtime.

\subsection{Pipeline Schedule}\label{sec:pipeline}

Pipeline schedule has long been known to improve bandwidth performance of collective communications. In a fully pipelined schedule, the bandwidth runtime is decided by the most congested link in the network. Formally speaking, if each link $e$ needs to deliver $m_e$ total amount of data in the schedule and has $b_e$ bandwidth, then the runtime of pipeline schedule is $\max_{e\in E}m_e/b_e$. Any step-based communication schedule can also be pipelined. The pipelined version always has bandwidth runtime no higher than the unpipelined one. \\
% Thus, by demonstrating our ability to generate optimal pipeline schedule, we show that our method generates schedules with optimal bandwidth performance among any type of schedules.\\

In a direct-connect topology where every node is a compute node, allgather pipeline schedule simply consists of a bunch of spanning out-trees from every node. However, this is no longer the case in switch topology. First of all, in switch topology, allgather only needs to broadcast from compute nodes and span compute nodes. Secondly, a spanning tree in switch topology may require switch nodes to broadcast data (figure \ref{fig:topo_tree}). While this is not a problem for programmable switches or switches supporting multicast, this limits the applicability of our work. In this paper, we utilize a graph theory technique called \textit{edge splitting} to convert the switch topology into a compute-node-only logical topology (figure \ref{fig:topo_split}). The generated pipeline schedule still achieves the optimal bandwidth runtime (\ref{eq:lowerbound}) in original switch topology. Thus, as an ancillary contribution, our results show that using switches capable of broadcast/reduction does not improve the bandwidth performance in mathematical cost model.

\subsection{Minimality-or-Saturation Dilemma}\label{sec:dilemma}

\begin{figure*}[t]
    \centering
    \begin{subfigure}{0.24\textwidth}
        \centering
        \scalebox{0.7}{
\begin{tikzpicture}[node/.style={rectangle,draw=black,minimum size=7mm,align=center}]
    \node[node]	(10)	at (-1.5,2.5-1) {$v^c_{1,1}$};
    \node[node]	(11)	at (-0.5,2.5-1) {$v^c_{1,2}$};
    \node[node]	(12)	at (0.5,2.5-1) {$v^c_{1,3}$};
    \node[node]	(13)	at (1.5,2.5-1) {$v^c_{1,4}$};
    \node[node,text width=3.5cm]	(14)	at (0,0.5+2.5) {Switch $v^s_1$};
    
    \path[latex-latex,anchor=east,line width=1pt] (10.north) edge node {$10b$} (10.north|-14.south);
    \path[latex-latex,line width=1pt] (11.north) edge (11.north|-14.south);
    \path[latex-latex,line width=1pt] (12.north) edge (12.north|-14.south);
    \path[latex-latex,line width=1pt] (13.north) edge (13.north|-14.south);
    
    \node[node]	(0)	at (-1.5,1-2.5) {$v^c_{2,1}$};
    \node[node]	(1)	at (-0.5,1-2.5) {$v^c_{2,2}$};
    \node[node]	(2)	at (0.5,1-2.5) {$v^c_{2,3}$};
    \node[node]	(3)	at (1.5,1-2.5) {$v^c_{2,4}$};
    \node[node,text width=3.5cm]	(4)	at (0,-0.5-2.5) {Switch $v^s_2$};
    
    \path[latex-latex,anchor=east,line width=1pt] (0.south) edge node {$10b$} (0.south|-4.north);
    \path[latex-latex,line width=1pt] (1.south) edge (1.south|-4.north);
    \path[latex-latex,line width=1pt] (2.south) edge (2.south|-4.north);
    \path[latex-latex,line width=1pt] (3.south) edge (3.south|-4.north);
            
    \node[node,text width=4cm]	(20)	at (0,0) {Switch $v^s_0$};
    
    \path[latex-latex,anchor=east] (0.north) edge node {$b$} (0.north|-20.south);
    \path[latex-latex] (1.north) edge (1.north|-20.south);
    \path[latex-latex] (2.north) edge (2.north|-20.south);
    \path[latex-latex] (3.north) edge (3.north|-20.south);
    
    \path[latex-latex,anchor=east] (10.south) edge node {$b$} (10.south|-20.north);
    \path[latex-latex] (11.south) edge (11.south|-20.north);
    \path[latex-latex] (12.south) edge (12.south|-20.north);
    \path[latex-latex] (13.south) edge (13.south|-20.north);
\end{tikzpicture}
}
        \caption{}
        \label{fig:topo}
    \end{subfigure}
    \begin{subfigure}{0.24\textwidth}
        \centering
        \scalebox{0.7}{
\begin{tikzpicture}[node/.style={rectangle,draw=black,minimum size=7mm,align=center}]
    \node[node]	(10)	at (-1.5,2.5-1) {$v^c_{1,1}$};
    \node[node]	(11)	at (-0.5,2.5-1) {$v^c_{1,2}$};
    \node[node]	(12)	at (0.5,2.5-1) {$v^c_{1,3}$};
    \node[node]	(13)	at (1.5,2.5-1) {$v^c_{1,4}$};
    \node[node,text width=3.5cm]	(14)	at (0,0.5+2.5) {Switch $v^s_1$};
    
    \path[latex-latex,line width=1pt] (10.north) edge (10.north|-14.south);
    \path[latex-latex,line width=1pt] (11.north) edge (11.north|-14.south);
    \path[latex-latex,line width=1pt] (12.north) edge (12.north|-14.south);
    \path[latex-latex,line width=1pt] (13.north) edge (13.north|-14.south);
    
    \node[node]	(0)	at (-1.5,1-2.5) {$v^c_{2,1}$};
    \node[node]	(1)	at (-0.5,1-2.5) {$v^c_{2,2}$};
    \node[node]	(2)	at (0.5,1-2.5) {$v^c_{2,3}$};
    \node[node]	(3)	at (1.5,1-2.5) {$v^c_{2,4}$};
    \node[node,text width=3.5cm]	(4)	at (0,-0.5-2.5) {Switch $v^s_2$};
    
    \path[latex-latex,line width=1pt] (0.south) edge (0.south|-4.north);
    \path[latex-latex,line width=1pt] (1.south) edge (1.south|-4.north);
    \path[latex-latex,line width=1pt] (2.south) edge (2.south|-4.north);
    \path[latex-latex,line width=1pt] (3.south) edge (3.south|-4.north);
            
    \node[node,text width=4cm]	(20)	at (0,0) {Switch $v^s_0$};
    
    \path[latex-latex] (0.north) edge (0.north|-20.south);
    \path[latex-latex] (1.north) edge (1.north|-20.south);
    \path[latex-latex] (2.north) edge (2.north|-20.south);
    \path[latex-latex] (3.north) edge (3.north|-20.south);
    
    \draw[line width=0.75pt, dashed] (-2.75,0.75) to[out=10, in=190] (2.75,0.75);
    \node[font=\Large] at (2.5, 1.25) {$S^*$};
    
    \path[latex-latex] (10.south) edge (10.south|-20.north);
    \path[latex-latex] (11.south) edge (11.south|-20.north);
    \path[latex-latex] (12.south) edge (12.south|-20.north);
    \path[latex-latex] (13.south) edge (13.south|-20.north);
\end{tikzpicture}
}
        \caption{}
        \label{fig:topo_partition}
    \end{subfigure}
    \begin{subfigure}{0.24\textwidth}
        \centering
        \scalebox{0.7}{
\begin{tikzpicture}[node/.style={rectangle,draw=black,minimum size=7mm,align=center}]
    \node[node, line width=1pt]	(10)	at (-1.5,2.5-1) {$v^c_{1,1}$};
    \node[node]	(11)	at (-0.5,2.5-1) {$v^c_{1,2}$};
    \node[node]	(12)	at (0.5,2.5-1) {$v^c_{1,3}$};
    \node[node]	(13)	at (1.5,2.5-1) {$v^c_{1,4}$};
    \node[node,text width=3.5cm]	(14)	at (0,0.5+2.5) {Switch $v^s_1$};
    
    \path[-latex] (10.north) edge (10.north|-14.south);
    \path[latex-] (11.north) edge (11.north|-14.south);
    \path[latex-] (12.north) edge (12.north|-14.south);
    \path[latex-] (13.north) edge (13.north|-14.south);
    
    \node[node]	(0)	at (-1.5,1-2.5) {$v^c_{2,1}$};
    \node[node]	(1)	at (-0.5,1-2.5) {$v^c_{2,2}$};
    \node[node]	(2)	at (0.5,1-2.5) {$v^c_{2,3}$};
    \node[node]	(3)	at (1.5,1-2.5) {$v^c_{2,4}$};
    \node[node,text width=3.5cm]	(4)	at (0,-0.5-2.5) {Switch $v^s_2$};
            
    \node[node,text width=4cm]	(20)	at (0,0) {Switch $v^s_0$};
    
    \path[latex-] (0.north) edge (0.north|-20.south);
    \path[latex-] (1.north) edge (1.north|-20.south);
    \path[latex-] (2.north) edge (2.north|-20.south);
    \path[latex-] (3.north) edge (3.north|-20.south);
    
    \path[-latex] (10.south) edge (10.south|-20.north);
\end{tikzpicture}
}
        \caption{}
        \label{fig:topo_tree}
    \end{subfigure}
    \begin{subfigure}{0.24\textwidth}
        \centering
        \scalebox{0.7}{
\begin{tikzpicture}[node/.style={rectangle,draw=black,minimum size=7mm,align=center}]
    \node[node]	(10)	at (-1.5,2.5-1) {$v^c_{1,1}$};
    \node[node]	(11)	at (-0.5,2.5-1) {$v^c_{1,2}$};
    \node[node]	(12)	at (0.5,2.5-1) {$v^c_{1,3}$};
    \node[node]	(13)	at (1.5,2.5-1) {$v^c_{1,4}$};

    \node[node,text width=3.5cm,draw=none]	(14)	at (0,0.5+2.5) {};

    \draw[-latex,line width=1pt] (10.north) to[out=70, in=110] (11.north);
    \draw[-latex,line width=1pt] (11.north) to[out=70, in=110] (12.north);
    \draw[-latex,line width=1pt] (12.north) to[out=70, in=110] (13.north);
    \draw[-latex, anchor=south,line width=1pt] (13.70) to[out=110, in=70] node {$10b$} (10.110);
    
    \node[node]	(0)	at (-1.5,1-2.5) {$v^c_{2,1}$};
    \node[node]	(1)	at (-0.5,1-2.5) {$v^c_{2,2}$};
    \node[node]	(2)	at (0.5,1-2.5) {$v^c_{2,3}$};
    \node[node]	(3)	at (1.5,1-2.5) {$v^c_{2,4}$};
    
    \node[node,text width=3.5cm,draw=none]	(4)	at (0,-0.5-2.5) {};
    
    \draw[-latex,line width=1pt] (0.south) to[out=-70, in=-110] (1.south);
    \draw[-latex,line width=1pt] (1.south) to[out=-70, in=-110] (2.south);
    \draw[-latex,line width=1pt] (2.south) to[out=-70, in=-110] (3.south);
    \draw[-latex, anchor=north,line width=1pt] (3.-70) to[out=-110, in=-70] node {$10b$} (0.-110);

    \draw[latex-] (0.north) to[out=70, in=110] (1.north);
    \draw[latex-] (1.north) to[out=70, in=110] (2.north);
    \draw[latex-] (2.north) to[out=70, in=110] (3.north);
    \draw[latex-] (10.-110) to[out=-110, in=110] (0.110);
    \draw[-latex] (10.south) to[out=-70, in=-110] (11.south);
    \draw[-latex] (11.south) to[out=-70, in=-110] (12.south);
    \draw[-latex] (12.south) to[out=-70, in=-110] (13.south);
    \draw[-latex,anchor=west] (13.-70) to[out=-70, in=70] node {$b$} (3.70);

    \draw[line width=0.75pt, dashed] (-2.75,0.5) to[out=10, in=190] (2.75,0.5);
    \node[font=\Large] at (2.5, 1) {$S^*$};  
\end{tikzpicture}
}
        \caption{}
        \label{fig:topo_split2}
    \end{subfigure}
    \caption{An 8-compute-node switch topology in 2-cluster setting. The thick links have 10x the bandwidth of the thin ones. Figure (a) shows the original switch topology. Figure (b) shows the bottleneck cut in this topology. Figure (c) shows a pipeline spanning tree rooted at $v_{1,1}^c$ with switch-node broadcast. Figure (d) shows a suboptimal way of transforming the switch topology into a direct-connect logical topology (resulting in 4x greater optimal runtime).}
\end{figure*}
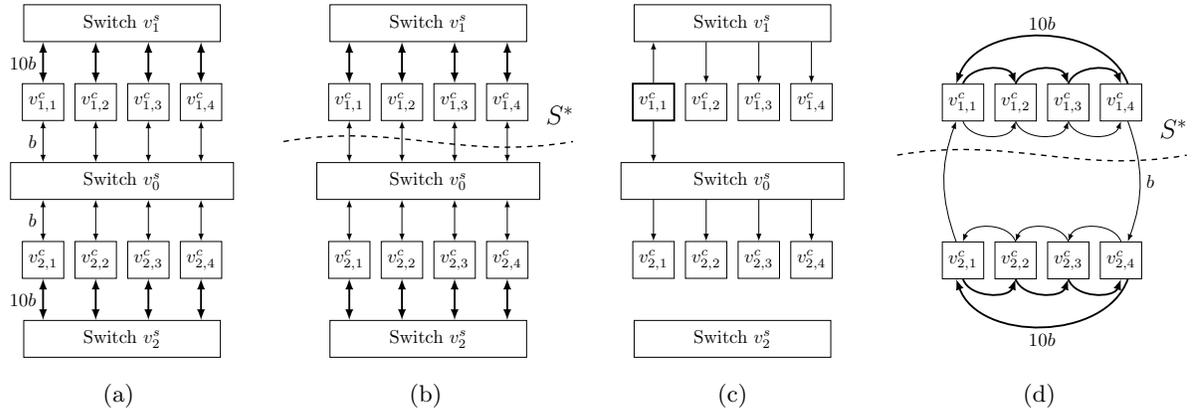

In this section, we discuss why we need a pipeline schedule instead of an ordinary step-based schedule to achieve bandwidth optimality. We show that in certain situations, pipeline schedule is \textit{the only possible way} to achieve bandwidth optimality. As shown in lower bound (\ref{eq:lowerbound}), the bandwidth performance of a topology is bounded by a bottleneck cut $(S^*,\overline{S^*})$. Suppose we want to achieve the performance bound given by the bottleneck cut, i.e. $(M/N)|S^*\cap V_c|/B_G^+(S^*)$, then the schedule must satisfy two requirements: (a) the bandwidth of the bottleneck cut, i.e. $B_G^+(S^*)$, must be saturated at all times, and (b) only the minimum amount of data required, i.e. $(M/N)|S^*\cap V_c|$, is transmitted through the bottleneck cut.\\

Consider the switch topology in figure \ref{fig:topo}. The topology has 8 compute nodes and 3 switch nodes. The eight compute nodes are in two clusters. Each cluster has a switch $v_1^s$ or $v_2^s$ providing $10b$ egress/ingress bandwidth for each compute node in the cluster. The 8 compute nodes are also connected to a global switch $v_0^s$, providing $b$ egress/ingress bandwidth for each compute node. It is easy to check that the bottleneck cut in this topology is a cluster cut $S^*=\{v_1^s,v_{1,1}^c,v_{1,2}^c,v_{1,3}^c,v_{1,4}^c\}$ shown in figure \ref{fig:topo_partition}. The cut has a runtime lower bound of $(M/N)(4/4b)$. In comparison, a single-compute-node cut has a much lower runtime lower bound $(M/N)(1/11b)$.\\

Suppose we want to achieve the lower bound by bottleneck cut $S^*$. Let $C$ be the last chunk sent through the cut to cluster 2, and suppose it is sent to $v_{2,1}^c$. The first thing to try is to saturate the bandwidth. It means that the schedule terminates right after $C$ is sent, leaving no idle time for $B_G^+(S^*)$. Then, at least one of $v_{2,2}^c,v_{2,3}^c,v_{2,4}^c$ must get $C$ directly from cluster 1 because they have no time to get it from $v_{2,1}^c$. This violates minimality, however, because chunk $C$ got sent through the bottleneck cut at least twice.\\

Suppose we want to achieve minimality. Then, $v_{2,1}^c$ has to broadcast $C$ to $v_{2,2}^c,v_{2,3}^c,v_{2,4}^c$ within the cluster. However, because $C$ is the last chunk sent through the cut by assumption, the cut bandwidth $B_G^+(S^*)$ is idle during the broadcast. The saturation requirement is violated. Thus, we are in a minimality-or-saturation dilemma that we cannot achieve both at the same time. However, we can do infinitely close by making chunk $C$ infinitesimally small. By doing so, we transmit minimum data required, and we also make the idle time of bottleneck cut close to 0. In step-based schedule, one always needs to specify $C$ as a fixed fraction of the total data, so it is impossible to achieve bandwidth optimality in such a case. In contrast, the size of one send/recv can be arbitrarily small in pipeline schedule. Therefore, pipeline schedule is the only way to achieve bandwidth optimality.

\section{Algorithm Design}

Let $G=(V=V_s\cup V_c,E)$ be an arbitrary network topology. We will compute an allgather pipeline schedule that reaches the lower bound (\ref{eq:lowerbound}) and is thus bandwidth optimal. We make two trivial assumptions about the topology: (a) all link bandwidths are integers and (b) $G$ is \textit{Eulerian} i.e. the total egress bandwidth equals the total ingress bandwidth for any node. For (a), when bandwidths are rational numbers, one can always scale them up to become integers. For (b), we use $B^+_G(v)$ and $B^-_G(v)$ to denote the total egress and ingress bandwidth of node $v$ respectively. Since $G$ is Eulerian, we have $B^+_G(v)=B^-_G(v)$ for all $v\in V$ and, consequently, $B^+_G(S)=B^-_G(S)$ for any $S\subseteq V$.\\

In summary, the algorithm contains three parts:
\begin{itemize}[leftmargin=*]
    \item \S\ref{sec:binarysearch}: Conduct a binary search to compute the lower bound (\ref{eq:lowerbound}), which is also the optimal bandwidth runtime. The binary search uses a network flow based oracle to test if a certain value is $\geq$ or $<$ than the true value of lower bound (\ref{eq:lowerbound}).
    \item \S\ref{sec:edgesplit}: Transform the switch topology into a direct-connect logical topology by using \textit{edge splitting} to remove switch nodes. The transformation is done without compromising optimal bandwidth performance. This part can be skipped if the input topology is already direct-connect.
    \item \S\ref{sec:constructtree}: Construct pipeline spanning trees in direct-connect topology to achieve optimal bandwidth performance. These spanning trees can then be mapped back to the original topology by reversing edge splitting, which determines the routing of communications between compute nodes.
\end{itemize}
The algorithm design is centered on earlier graph theoretical results on constructing edge-disjoint out-trees in directed graph~\cite{bang-jensen, tarjan, edmonds, berczi2010packing, schrijver}. A key observation leading to this algorithm is that \textit{given a set of out-trees, there are at most $U$ out-trees congested on any edge of $G$, if and only if, the set of out-trees is edge-disjoint in a multigraph topology obtained by duplicating each of $G$'s edges $U$ times.} \\

Another core design of our algorithm relies on \textit{edge splitting}, also a technique from graph theory~\cite{bang-jensen,frank,jackson}. It is used to transform the switch topology into a direct-connect topology so that one can construct compute-node-only pipeline spanning trees. Previous works such as TACCL~\cite{taccl} and TACOS~\cite{tacos} attempt to do this by ``unwinding'' switch topologies into predefined logical topologies, such as rings. However, their transformations often result in a loss of performance compared to the original switch topology. For example, the previous works may be tempted to unwind all switches in figure \ref{fig:topo} into rings, resulting in figure \ref{fig:topo_split2}. However, it makes the bottleneck cut $S^*$ worse that the egress bandwidth of $S^*$ becomes $b$ instead of $4b$, causing lower bound (\ref{eq:lowerbound}) being $(M/N)(4/b)$ (4x worse). In contrast, our \textit{edge splitting} strategically removes switch nodes without sacrificing any overall performance. Our transformation generates direct-connect topology in figure \ref{fig:topo_split}, which has the same runtime lower bound as figure \ref{fig:topo}. Furthermore, TACCL and TACOS are unable to handle topologies with multiple switches connected together like fat-tree~\cite{fattree} and dragonfly~\cite{dragonfly}, which are often the most popular ones in datacenters and high-performance computing.

\paragraph{Notation}In this paper, we make extensive use of network flow between different pairs of nodes. For any flow network $D$, we use $F(x,y;D)$ to denote the value of maxflow from $x$ to $y$ in $D$. For disjoint $A,B$, let $c(A,B;D)$ be the total capacity from $A$ to $B$ in $D$. By min-cut theorem, $F(x,y;D)\leq c(A,\bar{A};D)$ if $x\in A,y\in\bar{A}$, and there exists an $x$-$y$ cut $(A^*,\overline{A^*})$ that $F(x,y;D)=c(A^*,\overline{A^*};D)$.

\subsection{Optimality Binary Search}\label{sec:binarysearch}

In this section, we will show a way to compute the lower bound (\ref{eq:lowerbound}). Let $\{b_e\}_{e\in E}$ be the link bandwidths of $G$. By assumption, $\{b_e\}_{e\in E}$ are in $\Z_+$ and represented as capacities of edges in $G$. For any $x\in\Q$, we define $\vec{G}_x$ to be the flow network that (a) a source node $s$ is added and (b) an edge $(s,u)$ is added with capacity $x$ for every vertex $u\in V_c$. Now, we have the following theorem:

\begin{restatable}{theorem}{thmbinarysearch}\label{thm:binarysearch}
    $\min_{v\in V_c} F(s,v;\vec{G}_x)\geq|V_c|x$ if and only if $1/x\geq\max_{S\subset V,S\not\supseteq V_c}|S\cap V_c|/B^+_G(S)$.
\end{restatable}

The implication of theorem \ref{thm:binarysearch} is that we can do a binary search to get $1/x^*=\max_{S\subset V,S\not\supseteq V_c}|S\cap V_c|/B^+_G(S)$. The following initial range is trivial
\[
    \frac{N-1}{\min_{v\in V_c}B^-_G(v)}\leq\max_{S\subset V,S\not\supseteq V_c}\frac{|S\cap V_c|}{B^+_G(S)}\leq N-1.
\]
The lower bound corresponds to a partition containing all nodes except the compute node with minimum ingress bandwidth. The upper bound is due to the fact that $|S\cap V_c|\leq N-1$ and $B^+_G(S)\geq 1$. Starting with the initial range, one can then continuously test if $\min_{v\in V_c} F(s,v;\vec{G}_x)\geq|V_c|x$ for some midpoint $x$ to do a binary search. To find the exact $1/x^*$, let $S^*=\argmax_{S\subset V,S\not\supseteq V_c}|S\cap V_c|/B^+_G(S)$, then $1/x^*$ equals a fractional number with $B^+_G(S^*)$ as its denominator. Observe that $|S^*\cap V_c|\leq N-1$ and $|S^*\cap V_c|/B^+_G(S^*)\geq(N-1)/\min_{v\in V_c}B^-_G(v)$, so $B^+_G(S^*)\leq\min_{v\in V_c}B^-_G(v)$. Therefore, the denominator of $1/x^*$ is bounded by $\min_{v\in V_c}B^-_G(v)$.  Now, we use the following proposition:
\begin{restatable}{proposition}{lmfraction}\label{lm:fraction}
    Given two unequal fractional numbers $a/b$ and $c/d$ with $a,b,c,d\in\Z_+$, if denominators $b,d\leq X$ for some $X\in\Z_+$, then $|a/b-c/d|\geq 1/X^2$.
\end{restatable}
The proposition implies that if $1/x^*=a/b$ for some $b\leq\min_{v\in V_c}B^-_G(v)$, then any $c/d\neq 1/x^*$ with $d\leq\min_{v\in V_c}B^-_G(v)$ satisfies $|c/d-1/x^*|\geq 1/\min_{v\in V_c}B^-_G(v)^2$. Thus, one can run binary search until the range is smaller than $1/\min_{v\in V_c}B^-_G(v)^2$. Then, $1/x^*$ can be computed exactly by finding the fractional number closest to the midpoint with a denominator not exceeding $\min_{v\in V_c}B^-_G(v)$. The latter can be done with the continued fraction algorithm or brute force search if $\min_{v\in V_c}B^-_G(v)$ is small.\\

At the point, we have already known the optimality of bandwidth runtime given a topology $G$. For the remainder of this section, we will show that there exists a family of spanning trees that achieves this optimality. First of all, we have assumed that $G$'s links have the set of bandwidths $\{b_e\}_{e\in E}$. For the simplicity of notation, we use $G(\{c_e\})$ to denote the same topology as $G$ but with the set of bandwidths $\{c_e\}_{e\in E}$ instead. $\vec{G}_x(\{c_e\})$ is also defined accordingly. When $\{c_e\}_{e\in E}$ are integers, we say a family of out-trees $\cF$ is \textit{edge-disjoint} in $G(\{c_e\})$ if the number of trees using any edge $e\in E$ is less than or equal to $c_e$ i.e. $\sum_{T\in\cF}\I[e\in T]\leq c_e$ for all $e\in E$. The intuition behind this edge-disjointness is that \textit{the integer capacity $c_e$ represents the number of multiedges from the tail to the head of $e$.}\\

Now, we find $U\in\Q,k\in\N$ such that $U/k=1/x^*$ and $Ub_e\in\Z_+$ for all $e\in E$. For simplicity of schedule, we want $k$ to be as small as possible. The following proposition shows how to find such $U,k$:
\begin{restatable}{proposition}{lmsmallestUk}\label{lm:smallestUk}
     Given $\{b_e\}_{e\in E}\subset\Z_+$ and $1/x^*\in\Q$, let $p/q$ be the simplest fractional representation of $1/x^*$ i.e. $p/q=1/x^*$ and $\gcd(p,q)=1$. Suppose $k\in\N$ is the smallest such that there exists $U\in\Q$ satisfying $U/k=1/x^*$ and $Ub_e\in\Z_+$ for all $e\in E$, then $U=p/\gcd(q,\{b_e\}_{e\in E})$ and $k=Ux^*$.
\end{restatable}
In figure \ref{fig:topo}'s example, we have $1/x^*=|S^*\cap V_c|/B_G^+(S^*)=4/4b=1/b$ and thus $U=1/b,k=1$.\\

Consider the digraph $G(\{Ub_e\})$. Each edge of $G(\{Ub_e\})$ has integer capacity. We will show that there exists a family of edge-disjoint out-trees $\{T_{u,i}\}_{u\in V_c,i\in[k]}$ in $G(\{Ub_e\})$ with $T_{u,i}$ rooted at $u$ and $\cV(T_{u,i})\supseteq V_c$. Here, $[k]=\{1,2,\dots,k\}$ and $\cV(T_{u,i})$ denotes the vertex set of $T_{u,i}$. We use the following theorem proven by Bang-Jensen et al.~\cite{bang-jensen}:
\begin{restatable}[Bang-Jensen et al.~\cite{bang-jensen}]{theorem}{thmbangjensentrees}\label{thm:bang-jensen}
    Let $n\geq 1$ and $D=(V,E)$ be a digraph with a special node $s$. Let $T'=\{v\ |\ v\in V-s,d^-(v)<d^+(v)\}$. If $\lambda(s,v;D)\geq n$ for all $v\in T'$, then there is a family $\cF$ of edge-disjoint out-trees rooted at $s$ such that every $v\in V$ belongs to at least $\min(n,\lambda(s,v;D))$ number of out-trees.
\end{restatable}
Because we see integer capacity as the number of multiedges, here, the total in-degree $d^-(v)$ and out-degree $d^+(v)$ are simply the total ingress and egress capacity of $v$ in $G(\{Ub_e\})$. $\lambda(x,y;D)$ denotes the edge-connectivity from $x$ to $y$ in $D$ i.e. $\lambda(x,y;D)=\min_{x\in A,y\in\bar{A}}c(A,\bar{A};D)$. By min-cut theorem, $\lambda(x,y;D)$ is also equal to the maxflow from $x$ to $y$. Theorem \ref{thm:bang-jensen} leads to the following:
\begin{restatable}{theorem}{thmexisttrees}\label{thm:existtrees}
    Given integer-capacity digraph $D=(V_s\cup V_c,E)$ and $k\in\N$, there exists a family of edge-disjoint out-trees $\{T_{u,i}\}_{u\in V_c,i\in[k]}$ in $D$ with $T_{u,i}$ rooted at $u$ and $\cV(T_{u,i})\supseteq V_c$ if and only if \mbox{$\min_{v\in V_c}F(s,v;\vec{D}_k)\geq|V_c|k$.}
\end{restatable}
Consider the flow network $\vec{G}_k(\{Ub_e\})$. It is trivial to see that each edge in $\vec{G}_k(\{Ub_e\})$ has exactly $U$ times the capacity as in $\vec{G}_{x^*}$, including the edges incident from $s$. Thus, we have
\[
    \min_{v\in V_c} F(s,v;\vec{G}_k(\{Ub_e\}))=U\cdot\min_{v\in V_c}F(s,v;\vec{G}_{x^*})\geq U\cdot|V_c|x^*=|V_c|k.
\]
By theorem \ref{thm:existtrees}, there exists a family of edge-disjoint out-trees $\{T_{u,i}\}_{u\in V_c,i\in[k]}$ in $G(\{Ub_e\})$ with $T_{u,i}$ rooted at $u$ and $\cV(T_{u,i})\supseteq V_c$. Observe that for any edge $e \in E$, at most $Ub_e$ number of trees from $\{T_{u,i}\}_{u\in V_c,i\in[k]}$ use edge $e$. For allgather, we make each tree broadcast $1/k$ of the root's data shard, then the bandwidth runtime is
\[
    T_B\leq\max_{e\in E}\frac{M}{Nk}\cdot\frac{Ub_e}{b_e}=\frac{M}{N}\cdot\frac{U}{k}=\frac{M}{N}\cdot\frac{1}{x^*}=\frac{M}{N}\max_{S\subset V,S\not\supseteq V_c}\frac{|S\cap V_c|}{B^+_G(S)}
\]
reaching the lower bound (\ref{eq:lowerbound}) given topology $G$.\\

At this point, one may be tempted to construct and use $\{T_{u,i}\}_{u\in V_c,i\in[k]}$ to perform allgather. However, because $T_{u,i}$ can be arbitrary tree in $G(\{Ub_e\})$, it may force switch nodes to broadcast like $v_0^s,v_1^s$ in figure \ref{fig:topo_tree}. In the following section, we introduce a way to remove switch nodes from $G(\{Ub_e\})$, while preserving the existence of out-trees with the same bandwidth runtime. Afterward, we construct out-trees in the compute-node-only topology and map the communications back to $G(\{Ub_e\})$. Thus, we are able to construct a pipeline schedule with the same optimal bandwidth performance but without switch-node broadcast.

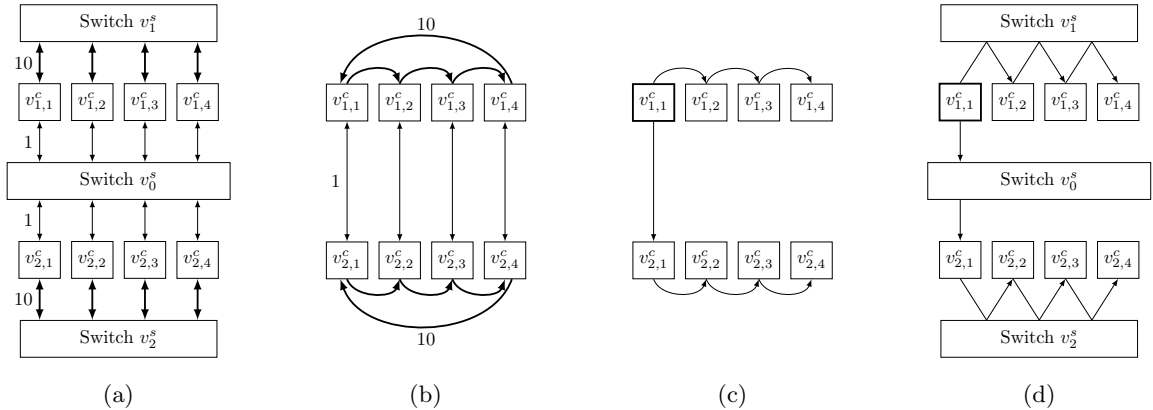
\begin{figure*}[t]
    \centering
    \begin{subfigure}{0.24\textwidth}
        \centering
        \scalebox{0.7}{
\begin{tikzpicture}[node/.style={rectangle,draw=black,minimum size=7mm,align=center}]
    \node[node]	(10)	at (-1.5,2.5-1) {$v^c_{1,1}$};
    \node[node]	(11)	at (-0.5,2.5-1) {$v^c_{1,2}$};
    \node[node]	(12)	at (0.5,2.5-1) {$v^c_{1,3}$};
    \node[node]	(13)	at (1.5,2.5-1) {$v^c_{1,4}$};
    \node[node,text width=3.5cm]	(14)	at (0,0.5+2.5) {Switch $v^s_1$};
    
    \path[latex-latex,anchor=east,line width=1pt] (10.north) edge node {$10$} (10.north|-14.south);
    \path[latex-latex,line width=1pt] (11.north) edge (11.north|-14.south);
    \path[latex-latex,line width=1pt] (12.north) edge (12.north|-14.south);
    \path[latex-latex,line width=1pt] (13.north) edge (13.north|-14.south);
    
    \node[node]	(0)	at (-1.5,1-2.5) {$v^c_{2,1}$};
    \node[node]	(1)	at (-0.5,1-2.5) {$v^c_{2,2}$};
    \node[node]	(2)	at (0.5,1-2.5) {$v^c_{2,3}$};
    \node[node]	(3)	at (1.5,1-2.5) {$v^c_{2,4}$};
    \node[node,text width=3.5cm]	(4)	at (0,-0.5-2.5) {Switch $v^s_2$};
    
    \path[latex-latex,anchor=east,line width=1pt] (0.south) edge node {$10$} (0.south|-4.north);
    \path[latex-latex,line width=1pt] (1.south) edge (1.south|-4.north);
    \path[latex-latex,line width=1pt] (2.south) edge (2.south|-4.north);
    \path[latex-latex,line width=1pt] (3.south) edge (3.south|-4.north);
            
    \node[node,text width=4cm]	(20)	at (0,0) {Switch $v^s_0$};
    
    \path[latex-latex,anchor=east] (0.north) edge node {$1$} (0.north|-20.south);
    \path[latex-latex] (1.north) edge (1.north|-20.south);
    \path[latex-latex] (2.north) edge (2.north|-20.south);
    \path[latex-latex] (3.north) edge (3.north|-20.south);
    
    \path[latex-latex,anchor=east] (10.south) edge node {$1$} (10.south|-20.north);
    \path[latex-latex] (11.south) edge (11.south|-20.north);
    \path[latex-latex] (12.south) edge (12.south|-20.north);
    \path[latex-latex] (13.south) edge (13.south|-20.north);
\end{tikzpicture}
}
        \caption{}
        \label{fig:topo2}
    \end{subfigure}
    \begin{subfigure}{0.24\textwidth}
        \centering
        \scalebox{0.7}{
\begin{tikzpicture}[node/.style={rectangle,draw=black,minimum size=7mm,align=center}]
    \node[node]	(10)	at (-1.5,2.5-1) {$v^c_{1,1}$};
    \node[node]	(11)	at (-0.5,2.5-1) {$v^c_{1,2}$};
    \node[node]	(12)	at (0.5,2.5-1) {$v^c_{1,3}$};
    \node[node]	(13)	at (1.5,2.5-1) {$v^c_{1,4}$};

    \node[node,text width=3.5cm,draw=none]	(14)	at (0,0.5+2.5) {};

    \draw[-latex,line width=1pt] (10.north) to[out=70, in=110] (11.north);
    \draw[-latex,line width=1pt] (11.north) to[out=70, in=110] (12.north);
    \draw[-latex,line width=1pt] (12.north) to[out=70, in=110] (13.north);
    \draw[-latex, anchor=south,line width=1pt] (13.70) to[out=110, in=70] node {$10$} (10.110);
    
    \node[node]	(0)	at (-1.5,1-2.5) {$v^c_{2,1}$};
    \node[node]	(1)	at (-0.5,1-2.5) {$v^c_{2,2}$};
    \node[node]	(2)	at (0.5,1-2.5) {$v^c_{2,3}$};
    \node[node]	(3)	at (1.5,1-2.5) {$v^c_{2,4}$};
    
    \node[node,text width=3.5cm,draw=none]	(4)	at (0,-0.5-2.5) {};
    
    \draw[-latex,line width=1pt] (0.south) to[out=-70, in=-110] (1.south);
    \draw[-latex,line width=1pt] (1.south) to[out=-70, in=-110] (2.south);
    \draw[-latex,line width=1pt] (2.south) to[out=-70, in=-110] (3.south);
    \draw[-latex, anchor=north,line width=1pt] (3.-70) to[out=-110, in=-70] node {$10$} (0.-110);
    
    \path[latex-latex, anchor=east] (0.north) edge node {$1$} (0.north|-10.south);
    \path[latex-latex] (1.north) edge (1.north|-11.south);
    \path[latex-latex] (2.north) edge (2.north|-12.south);
    \path[latex-latex] (3.north) edge (3.north|-13.south);
\end{tikzpicture}
}
        \caption{}
        \label{fig:topo_split}
    \end{subfigure}
    \begin{subfigure}{0.24\textwidth}
        \centering
        \scalebox{0.7}{
\begin{tikzpicture}[node/.style={rectangle,draw=black,minimum size=7mm,align=center}]
    \node[node, line width=1pt]	(10)	at (-1.5,2.5-1) {$v^c_{1,1}$};
    \node[node]	(11)	at (-0.5,2.5-1) {$v^c_{1,2}$};
    \node[node]	(12)	at (0.5,2.5-1) {$v^c_{1,3}$};
    \node[node]	(13)	at (1.5,2.5-1) {$v^c_{1,4}$};

    \node[node,text width=3.5cm,draw=none]	(14)	at (0,0.5+2.5) {};

    \draw[-latex] (10.north) to[out=70, in=110] (11.north);
    \draw[-latex] (11.north) to[out=70, in=110] (12.north);
    \draw[-latex] (12.north) to[out=70, in=110] (13.north);
    \draw[-latex, draw=none] (13.70) to[out=110, in=70] (10.110);
    
    \node[node]	(0)	at (-1.5,1-2.5) {$v^c_{2,1}$};
    \node[node]	(1)	at (-0.5,1-2.5) {$v^c_{2,2}$};
    \node[node]	(2)	at (0.5,1-2.5) {$v^c_{2,3}$};
    \node[node]	(3)	at (1.5,1-2.5) {$v^c_{2,4}$};

    \node[node,text width=3.5cm,draw=none]	(4)	at (0,-0.5-2.5) {};
    
    \draw[-latex] (0.south) to[out=-70, in=-110] (1.south);
    \draw[-latex] (1.south) to[out=-70, in=-110] (2.south);
    \draw[-latex] (2.south) to[out=-70, in=-110] (3.south);
    \draw[-latex, draw=none] (3.-70) to[out=-110, in=-70] (0.-110);
    
    \path[latex-] (0.north) edge (0.north|-10.south);
\end{tikzpicture}
}
        \caption{}
        \label{fig:topo_split_tree}
    \end{subfigure}
    \begin{subfigure}{0.24\textwidth}
        \centering
        \scalebox{0.7}{
\begin{tikzpicture}[node/.style={rectangle,draw=black,minimum size=7mm,align=center}]
    \node[node, line width=1pt]	(10)	at (-1.5,2.5-1) {$v^c_{1,1}$};
    \node[node]	(11)	at (-0.5,2.5-1) {$v^c_{1,2}$};
    \node[node]	(12)	at (0.5,2.5-1) {$v^c_{1,3}$};
    \node[node]	(13)	at (1.5,2.5-1) {$v^c_{1,4}$};
    \node[node,text width=3.5cm]	(14)	at (0,0.5+2.5) {Switch $v^s_1$};
    
    \draw (10.north) -- ($(14.south -| -1, 0)$);
    \draw[-latex] ($(14.south -| -1, 0)$) -- (11.north);
    
    \draw (11.north) -- ($(14.south -| 0, 0)$);
    \draw[-latex] ($(14.south -| 0, 0)$) -- (12.north);
    
    \draw (12.north) -- ($(14.south -| 1, 0)$);
    \draw[-latex] ($(14.south -| 1, 0)$) -- (13.north);

    \node[node]	(0)	at (-1.5,1-2.5) {$v^c_{2,1}$};
    \node[node]	(1)	at (-0.5,1-2.5) {$v^c_{2,2}$};
    \node[node]	(2)	at (0.5,1-2.5) {$v^c_{2,3}$};
    \node[node]	(3)	at (1.5,1-2.5) {$v^c_{2,4}$};
    \node[node,text width=3.5cm]	(4)	at (0,-0.5-2.5) {Switch $v^s_2$};
            
    \node[node,text width=4cm]	(20)	at (0,0) {Switch $v^s_0$};
    
    \path[latex-] (0.north) edge (0.north|-20.south);
    \path[-latex] (10.south) edge (10.south|-20.north);
    
    \draw (0.south) -- ($(4.north -| -1, 0)$);
    \draw[-latex] ($(4.north -| -1, 0)$) -- (1.south);
    
    \draw (1.south) -- ($(4.north -| 0, 0)$);
    \draw[-latex] ($(4.north -| 0, 0)$) -- (2.south);
    
    \draw (2.south) -- ($(4.north -| 1, 0)$);
    \draw[-latex] ($(4.north -| 1, 0)$) -- (3.south);
\end{tikzpicture}
}
        \caption{}
        \label{fig:topo_tree2}
    \end{subfigure}
    \caption{Different stages of the topology in schedule construction. Figure (a) shows the topology of $G(\{Ub_e\})$. Note that the link capacities no longer have $b$ as a multiplier. Figure (b) shows the topology $G^*$ after edge splitting removes all switch nodes. Figure (c) shows a pipeline spanning tree constructed in $G^*$. Figure (d) shows the routings in $G$ corresponding to the spanning tree.}
    \label{fig:states}
\end{figure*}

\subsection{Edge Splitting}\label{sec:edgesplit}

\begin{algorithm}[tb]
    \caption{Remove Switch Nodes}
    \label{algo:removeswitch}
    \SetAlgoLined
    \DontPrintSemicolon
    \KwIn{Integer-capacity Eulerian digraph $D=(V_s\cup V_c,E)$ and $k\in\N$.}
    \KwOut{Direct-connect digraph $D^*=(V_c,E^*)$ and path recovery table \emap.}
    \Begin{
        Initialize table \emap\;
        \ForEach{\normalfont switch node $w\in V_s$}{
            \ForEach{\normalfont egress edge $f=(w,t)\in E$}{
                \ForEach{\normalfont ingress edge $e=(u,w)\in E$}{
                    % \tcp{One may prioritize some edge with $u,t$ in distinct clusters.}
                    Compute $M$ as in (\ref{eq:maxsplitoff}).\;
                    \If{$M>0$}{
                        Decrease $f$'s and $e$'s capacity by $M$. Remove $e$ if its capacity reaches 0.\;
                        Increase capacity of $(u,t)$ by $M$. Add the edge if $(u,t)\notin E$.\;
                        $\emap[(u,t)][w]\leftarrow\emap[(u,t)][w]+M$\;
                        \lIf{\normalfont $f$'s capacity reaches 0}{\textbf{break}}
                    }
                }
                \tcp{Edge $f$ should have 0 capacity at this point.}
                Remove edge $f$ from $D$.\;
            }
            \tcp{Node $w$ should be isolated at this point.}
            Remove node $w$ from $D$.\;
        }
        \Return{\normalfont the latest $D$ as $D^*$ and table \emap}
    }
\end{algorithm}

To remove the switch nodes from $G(\{Ub_e\})$, we apply a technique called \textit{edge splitting}. Consider a vertex $w$ and two incident edges $(u,w),(w,t)$. The operation of edge splitting is to replace $(u,w),(w,t)$ by a direct edge $(u,t)$ while maintaining edge-connectivities in the graph. In our context, $w$ is a switch node. We continuously split off one capacity of an incoming edge to $w$ and one capacity of an outgoing edge from $w$ until $w$ is isolated and can be removed from the graph. Because the edge-connectivities are maintained, we are able to show that $\min_{v\in V_c} F(s,v;\vec{G}_k(\{Ub_e\}))\geq|V_c|k$ is maintained in the process. Thus, by theorem \ref{thm:existtrees}, the existence of spanning trees with the same optimal bandwidth performance is also preserved.\\

We start with the following theorem from Bang-Jensen et al.~\cite{bang-jensen}. The theorem was originally proven by Frank~\cite{frank} and Jackson~\cite{jackson}.
\begin{restatable}[Bang-Jensen et al.~\cite{bang-jensen}]{theorem}{thmbangjensenedgesplit}\label{thm:BJedgesplitting}
    Let $D=(V+w,E)$ be a directed Eulerian graph, that is, $d^-(x)=d^+(x)$ for every node $x$ of $D$. Then, for every edge $f=(w,t)$ there is an edge $e=(u,w)$ such that $\lambda(x,y;D^{ef})=\lambda(x,y;D)$ for every $x,y\in V$, where $D^{ef}$ is the resulting graph obtained by splitting off $e$ and $f$ in $D$.
\end{restatable}
In our case, we are not concerned with any edge-connectivity other than from $s$. We allow $\lambda(x,y;D^{ef})\neq\lambda(x,y;D)$ as long as $\min_{v\in V_c}F(s,v;\vec{D}^{ef}_{k})=\min_{v\in V_c}\lambda(s,v;\vec{D}^{ef}_{k})\geq|V_c|k$ holds after splitting. Theorem \ref{thm:BJedgesplitting} is used to derive the following theorem:
\begin{restatable}{theorem}{thmedgesplit}\label{thm:edgesplit}
    Given integer-capacity Eulerian digraph $D=(V_s\cup V_c,E)$ and $k\in\N$ with $\min_{v\in V_c}F(s,v;\vec{D}_k)\geq|V_c|k$, for every edge $f=(w,t)$ $(w\in V_s)$ there is an edge $e=(u,w)$ such that $\min_{v\in V_c}F(s,v;\vec{D}^{ef}_{k})\geq|V_c|k$.
\end{restatable}
Note that here, $f$ and $e$ each represent one of the multiedges (or one capacity) between $w,t$ and $u,w$, respectively. Observe that edge splitting does not affect a graph being Eulerian. Thus, in $G(\{Ub_e\})$, we can iteratively replace edges $e=(u,w),f=(w,t)$ by $(u,t)$ for each switch node $w\in V_s$, while maintaining $\min_{v\in V_c}F(s,v;\vec{G}^{ef}_{k}(\{Ub_e\}))\geq|V_c|k$. The resulting graph will have all nodes in $V_s$ isolated. By removing $V_s$, we get a graph $G^*=(V_c,E^*)$ having compute nodes only. Because of theorem \ref{thm:existtrees}, there exists a family of edge-disjoint out-trees in $G^*=(V_c,E^*)$ that achieves the optimal bandwidth performance.\\

While one can split off one capacity of $(u,w),(w,t)$ at a time, this becomes inefficient if the capacities of edges are large. Here, we introduce a way to split off $(u,w),(w,t)$ by maximum capacity at once. Given edges $(u,w),(w,t)\in E$, we construct a flow network $\widehat{D}_{(u,w),v}$ from $\vec{D}_k$ for each $v\in V_c$ that $\widehat{D}_{(u,w),v}$ connects $(u,s),(u,t),(v,w)$ with $\infty$ capacity. Similarly, we construct a flow network $\widehat{D}_{(w,t),v}$ that connects $(w,s),(u,t),(v,t)$ with $\infty$ capacity.
\begin{restatable}{theorem}{thmmulticapasplit}\label{thm:multicapasplit}
    Given integer-capacity Eulerian digraph $D=(V_s\cup V_c,E)$ and $k\in\N$ with $\min_{v\in V_c}F(s,v;\vec{D}_k)\geq|V_c|k$, the maximum capacity that $e=(u,w),f=(w,t)$ can be splitted off with the resulting graph $D^{ef}$ satisfying $\min_{v\in V_c}F(s,v;\vec{D}^{ef}_{k})\geq|V_c|k$ is
    \begin{equation}\label{eq:maxsplitoff}
        M=\min\left\{c(u,w;D)\ ,\ c(w,t;D)\ ,\ \min_{v\in V_c} F(u,w;\widehat{D}_{(u,w),v})-|V_c|k\ ,\ \min_{v\in V_c} F(w,t;\widehat{D}_{(w,t),v})-|V_c|k\right\}.
    \end{equation}
\end{restatable}
Based on theorem \ref{thm:multicapasplit}, we are able to develop algorithm \ref{algo:removeswitch}. What is remarkable about algorithm \ref{algo:removeswitch} is that its runtime does not depend on the capacities of the digraph. One should also note that we update a table \emap while splitting. After edge splitting, we are ready to construct spanning trees that only use compute nodes for broadcast. \emap is then used to convert the spanning trees back to paths in $G$ that use switch nodes for send/receive between compute nodes.\\

Figure \ref{fig:states} gives an example of edge splitting. In figure \ref{fig:topo2}, within each cluster $i\in\{1,2\}$, we split off 10 capacity of $(v^c_{i,j},v^s_{i}),(v^s_{i},v^c_{i,(j\bmod 4)+1})$ for $j=1,2,3,4$ to form a ring topology. Across clusters, we split off $1$ capacity of $(v^c_{i,j},v^s_0),(v^s_0,v^c_{(i\bmod 2)+1,j})$ for $j=1,2,3,4$. The resulting topology figure \ref{fig:topo_split} has compute nodes only, and the optimal bandwidth runtime is still $(M/N)(4/4b)$ if bandwidth multiplier $b$ is added. Note that for this example, in the innermost foreach loop of algorithm \ref{algo:removeswitch}, we adjusted the order of iterating through $e$s to prioritize splitting off $(u,v_0^s),(v_0^s,t)$ pairs with $u,t$ in different clusters. The adjustment is not for performance-related reasons, but rather to simplify routing by scheduling all intra-cluster traffic through the in-cluster switch. We successfully met this goal: the capacity of each $f$ reaches 0 before we iterate to an $e$ with $u,t$ in the same cluster.

\subsection{Spanning Tree Construction}\label{sec:constructtree}

\begin{algorithm}[tb]
    \caption{Spanning Tree Construction}
    \label{algo:construct}
    \SetAlgoLined
    \DontPrintSemicolon
    \KwIn{Integer-capacity digraph $D^*=(V_c,E^*)$ and $k\in\N$.}
    \KwOut{Spanning tree $(R_{u,i},\cE(R_{u,i}))$ for each $u\in V_c,i\in[n_u]$. Subgraph $(R_{u,i},\cE(R_{u,i}))$s satisfy $\forall u\in V_c:\sum_{i=1}^{n_u}m(R_{u,i})=k$ and $\forall e\in E^*:\sum\{m(R_{u,i})\ |\ e\in \cE(R_{u,i})\}\leq c(e;D^*)$.}
    \Begin{
        Initialize $R_{u,1}=\{u\},\cE(R_{u,1})=\emptyset,m(R_{u,1})=k,n_u=1$ for all $u\in V_c$.\;
        Initialize $g(e)=c(e;D^*)$ for all $e\in E^*$.\;
        \While{\normalfont there exists $R_{u,i}\neq V_c$}{
            \While{\normalfont $R_{u,i}\neq V_c$}{
                Pick an edge $(x,y)$ in $D^*$ that $x\in R_{u,i},y\notin R_{u,i}$.\;
                Compute $\mu$ as in (\ref{eq:computemu}).\;
                \lIf{\normalfont $\mu=0$}{\textbf{continue}}
                \If{\normalfont $\mu<m(R_{u,i})$}{
                    $n_{u}\leftarrow n_{u}+1$\;
                    Create a new copy $R_{u,n_{u}}=R_{u,i},\cE(R_{u,n_{u}})=\cE(R_{u,i}),m(R_{u,n_{u}})=m(R_{u,i})-\mu$.\;
                    $m(R_{u,i})\leftarrow\mu$\;
                }
                $\cE(R_{u,i})\leftarrow\cE(R_{u,i})+(x,y)$\;
                $R_{u,i}\leftarrow R_{u,i}+y$\;
                $g(x,y)\leftarrow g(x,y)-\mu$. Remove $(x,y)$ if $g(x,y)$ reaches $0$.\;
            }
        }
    }
\end{algorithm}

At this point, we have a digraph $G^*=(V_c,E^*)$ with only compute nodes. In this section, we construct $k$ out-trees from every node that span all nodes $V_c$ in $G^*$. We start by showing the existence of spanning trees with the following theorem in Tarjan~\cite{tarjan}. The theorem was originally proven by Edmonds~\cite{edmonds}.
\begin{restatable}[Tarjan~\cite{tarjan}]{theorem}{thmtarjan}\label{thm:tarjan}
    For any integer-capacity digraph $D=(V,E)$ and any sets $R_i\subseteq V$, $i\in[k]$, there exist $k$ edge-disjoint spanning out-trees $T_i$, $i\in[k]$, rooted respectively at $R_i$, if and only if for every $S\neq V$,   
    \begin{equation}\label{eq:tjcondition}
        c(S,\bar{S};D)\geq|\{i\ |\ R_i\subseteq S\}|.
    \end{equation}
\end{restatable}
A spanning out-tree is \textit{rooted at $R_i$} if for every $v\in V-R_i$, there is exactly one directed path from a vertex in $R_i$ to $v$ within the acyclic subgraph of out-tree. To see there exists a family of edge-disjoint spanning out-trees $\{T_{u,i}\}_{u\in V_c,i\in[k]}$ in $G^*$, observe that each $T_{u,i}$ is rooted at $R_{u,i}=\{u\}$, so $|\{(u,i)\ |\ R_{u,i}\subseteq S\}|=|S|k$ for any $S\subset V_c,S\neq V_c$. We show the following theorem:
\begin{restatable}{theorem}{thminitialexist}\label{thm:initialexist}
    Given integer-capacity digraph $D=(V_c,E)$ and $k\in\N$, $c(S,\bar{S};D)\geq|S|k$ for all $S\subset V_c,S\neq V_c$ if and only if $\min_{v\in V_c}F(s,v;\vec{D}_k)\geq|V_c|k$.
\end{restatable}
Since we ensured $\min_{v\in V_c}F(s,v;\vec{G}^*_k)\geq|V_c|k$, condition (\ref{eq:tjcondition}) is satisfied. Spanning tree construction essentially involves iteratively expanding each $R_{u,i}=\cV(T_{u,i})$ from $\{u\}$ to $V_c$ by adding edges to $T_{u,i}$, while maintaining condition (\ref{eq:tjcondition}). Tarjan~\cite{tarjan} has proposed such an algorithm. For each $T_{u,i}$, the algorithm continuously finds an edge $(x,y)$ with $x\in R_{u,i},y\notin R_{u,i}$ that adding this edge to $T_{u,i}$ does not violate (\ref{eq:tjcondition}). It is proven that such an edge is guaranteed to exist. However, the runtime of the algorithm quadratically depends on the total number of spanning trees, i.e. $Nk$ in our case. This becomes problematic when $k$ is large, as $k$ can get up to $\min_{v\in V_c}B^-_G(v)/\gcd(\{b_e\}_{e\in E})$. Fortunately, B{\'e}rczi \& Frank~\cite{berczi2010packing} has proposed a strongly polynomial-time algorithm based on Schrijver~\cite{schrijver}. The runtime of the algorithm does not depend on $k$ at all. The following theorem has been shown:
\begin{restatable}[B{\'e}rczi \& Frank~\cite{berczi2010packing}]{theorem}{thmpolynomialpacking}\label{thm:polynomialpacking}
    Let $D=(V,E)$ be a digraph, $g:E\to\Z_+$ a capacity function, $\mathcal{R}=\{R_1,\dots,R_n\}$ a list of root-sets, $\mathcal{U}=\{U_1,\dots,U_n\}$ a set of convex sets with $R_i\subseteq U_i$, and $m:\mathcal{R}\to\Z_+$ a demand function. There is a strongly polynomial time algorithm that finds (if there exist) $m(\mathcal{R})$ out-trees so that $m(R_i)$ of them are spanning $U_i$ with root-set $R_i$ and each edge $e\in E$ is contained in at most $g(e)$ out-trees.
\end{restatable}
In our context, we start with $\mathcal{R}=\{R_u\ |\ u\in V_c\}$ and $R_u=\{u\},U_u=V_c,m(R_u)=k$. We define $\cE(R_i)$ to be the edge set of the $m(R_i)$ out-trees corresponding to $R_i$, so $\cE(R_u)=\emptyset$ is initialized. Given $\mathcal{R}=\{R_1,\dots,R_n\}$, we pick an $R_i\neq V_c$, say $R_1$. Then, we find an edge $(x,y)$ such that $x\in R_1,y\notin R_1$ and $(x,y)$ can be added to $\mu:0<\mu\leq\min\{g(x,y),m(R_1)\}$ copies of the $m(R_1)$ out-trees without violating (\ref{eq:tjcondition}). If $\mu=m(R_1)$, then we directly add $(x,y)$ to $\cE(R_1)$ and $R_1=R_1+y$. If $\mu<m(R_1)$, then we add a copy $R_{n+1}$ of $R_1$ that $\cE(R_{n+1})=\cE(R_1),m(R_{n+1})=m(R_1)-\mu$. We revise $m(R_1)$ to $\mu$, add $(x,y)$ to $\cE(R_1)$, and $R_1=R_1+y$. Finally, we update $g(x,y)=g(x,y)-\mu$. Now, given $\mathcal{R}=\{R_1,\dots,R_{n+1}\}$, we can apply the step recursively until $R_i=V_c$ for all $R_i\in\mathcal{R}$. According to B{\'e}rczi \& Frank~\cite{berczi2010packing}, $\mu$ is defined as followed:
\[
    \textstyle\mu=\min\left\{g(x,y)\ ,\ m(R_1)\ ,\ \min\{c(S,\bar{S};D)-p(S;D):x\in S,y\in\bar{S},R_1\not\subseteq S\}\right\}
\]
where $p(S;D)=\sum\{m(R_i)\ |\ R_i\subseteq S\}$. Neither B{\'e}rczi \& Frank~\cite{berczi2010packing} nor Schrijver~\cite{schrijver} explicitly mentioned how to compute $\mu$ in polynomial time. Therefore, we describe a method for doing so. We construct a flow network $\overline{D}$ such that (a) a node $s_i$ is added for each $R_i$ except $i=1$, (b) connect $x$ to each $s_i$ with capacity $m(R_i)$, and (c) connect each $s_i$ to every vertex in $R_i$ with $\infty$ capacity. We show the following result:
\begin{restatable}{theorem}{thmcomputemu}\label{thm:computemu}
    For any edge $(x,y)$ in $D$ with $x\in R_1,y\notin R_1$,
    \begin{equation}\label{eq:computemu}
        \textstyle\mu=\min\left\{g(x,y)\ ,\ m(R_1)\ ,\ F(x,y;\overline{D})-\sum_{i\neq 1}m(R_i)\right\}.
    \end{equation}
\end{restatable}
Thus, $\mu$ can be calculated by computing a single maxflow from $x$ to $y$ in $\overline{D}$. The complete algorithm is described in algorithm \ref{algo:construct}. The resulting $\mathcal{R}$ can be indexed as $\mathcal{R}=\bigcup_{u\in V_c}\{R_{u,1},\dots,R_{u,n_{u}}\}$, where $R_{u,i}$ corresponds to $m(R_{u,i})$ number of identical out-trees rooted at $u$ and specified by edge set $\cE(R_{u,i})$. We have $\sum_{i=1}^{n_u}m(R_{u,i})=k$ for all $u$. Thus, $\mathcal{R}$ can be decomposed into $\{T_{u,i}\}_{u\in V_c,i\in[k]}$. However, since all spanning trees within $R_{u,i}$ are identical, the allgather schedule can simply be specified in terms of $\cE(R_{u,i})$ and $m(R_{u,i})$.\\

After construction, we have edge-disjoint spanning trees $\{T_{u,i}\}_{u\in V_c,i\in [k]}$ in $G^*$. Each of the edge $(u,v)$ in $T_{u,i}$ may correspond to a path $u\to w_1\to\dots\to w_n\to v$ in $G$ with $w_1,\dots,w_n$ being switch nodes. In other words, edges in $T_{u,i}$ only specify the source and destination of send/recv between compute nodes. Thus, one needs to use the \emap in algorithm \ref{algo:removeswitch} to recover the paths in $G$. For any edge $(u,t)$ in $G^*$, $\emap[(u,t)][w]$ denotes the amount of capacity from $u$ to $t$ that is going through $(u,w),(w,t)$. It should be noted that \emap may be recursive, meaning that $(u,w),(w,t)$ may also go through some other switches. Because each capacity of $(u,t)$ corresponds to one capacity of a path from $u$ to $t$ in $G$, the resulting pipeline schedule in $G$ has the same performance in $G^*$, achieving the optimal bandwidth performance (\ref{eq:lowerbound}) in $G$. \\

In figure \ref{fig:states}'s example, we construct a spanning tree like \ref{fig:topo_split_tree} for each of the compute node. By reversing the edge splitting with \emap, the spanning tree becomes the pipeline schedule in \ref{fig:topo_tree2}. Note that the corresponding pipeline schedule of a spanning tree in $G^*$ is not necessarily a tree in $G$. For example, the pipeline schedule in \ref{fig:topo_tree2} visits switches $v_1^s,v_2^s$ multiple times. To obtain a complete allgather pipeline schedule with optimal bandwidth runtime $(M/N)(4/4b)$, one can apply the similar pipeline schedule for each of the compute nodes in \ref{fig:topo_tree2}.\\

One may be tempted to devise a way to construct spanning trees with low heights. This has numerous benefits such as lower latency at small data sizes and better convergence of pipeline schedule towards bandwidth optimality. Although there is indeed potential progress to be made in this direction, constructing edge-disjoint spanning trees of minimum height has been proven to be NP-complete~\cite{npcomplete}.

\subsection[Fixed-k Optimality]{Fixed-$k$ Optimality}\label{sec:fixedk}

A potential problem of our pipeline schedule is that $k$, the number of spanning trees per root, depends linearly on link bandwidths, potentially reaching up to $\min_{v\in V_c}B^-_G(v)/\gcd(\{b_e\}_{e\in E})$. Although the runtime of spanning tree construction does not depend on $k$, in practice, one may want to reduce $k$ to simplify the pipeline schedule. In this section, we offer a way to construct a pipeline schedule with the best possible bandwidth performance for a fixed $k$. We start with the following theorem:

% A potential problem of constructing bandwidth optimal spanning trees is that $k$, the number of spanning trees per root, depends linearly on link bandwidths, potentially reaching up to $\min_{v\in V_c}B^-_G(v)/\gcd(\{b_e\}_{e\in E})$. Although in practice, we have never encountered runtime issues with the batched spanning tree construction, theoretically, $k$ can be intractably large in extreme cases, such as $\{b_e\}_{e\in E}$ are all very large prime numbers. In such cases, one may want to compute the best possible bandwidth performance given a fixed small $k$. There may also be a desire to reduce $k$ to simplify the pipeline schedule. In this section, we offer a way to construct a pipeline schedule with the best possible bandwidth performance for a fixed $k$. We start with the following theorem:

\begin{restatable}{theorem}{thmfixedk}\label{thm:fixedk}
    Given $U\in\R_{+}$ and $k\in\N$, a family of out-trees $\{T_{u,i}\}_{u\in V_c,i\in[k]}$ with $T_{u,i}$ rooted at $u$ and $\cV(T_{u,i})\supseteq V_c$ achieves $\frac{M}{Nk}\cdot U$ bandwidth runtime if and only if it is edge-disjoint in $G(\{\lfloor Ub_e\rfloor\}_{e\in E})$.
\end{restatable}

To test the existence of edge disjoint $\{T_{u,i}\}_{u\in V_c,i\in[k]}$ in $G(\{\lfloor Ub_e\rfloor\}_{e\in E})$, by theorem \ref{thm:existtrees}, we can simply test whether $\min_{v\in V_c}F(s,v;\vec{G}_k(\{\lfloor Ub_e\rfloor\}))\geq|V_c|k$ holds. The following theorem provides a method for binary search to find the lowest bandwidth runtime for the given $k$.

\begin{restatable}{theorem}{thmfixedkbinarysearch}\label{thm:fixedkbinarysearch}
    Let $\frac{M}{Nk}\cdot U^*$ be the lowest bandwidth runtime that can be achieved with $k$ out-trees per $v\in V_c$. Then, there exists a family of edge-disjoint out-trees $\{T_{u,i}\}_{u\in V_c,i\in[k]}$ in $G(\{\lfloor Ub_e\rfloor\}_{e\in E})$ with $T_{u,i}$ rooted at $u$ and $\cV(T_{u,i})\supseteq V_c$ if and only if $U\geq U^*$.
\end{restatable}
The initial range is:
\[
    \frac{(N-1)k}{\min_{v\in V_c}B^-_G(v)}\leq U^*\leq (N-1)k.
\]
Observe that there must exists $b_e\in E$ such that $U^*b_e\in\Z_+$; otherwise, $U^*$ can be further decreased. Thus, the denominator of $U^*$ must be less than or equal to $\max_{e\in E}b_e$. Similar to optimality binary search, by proposition \ref{lm:fraction}, one can run binary search until the range is smaller than $1/\max_{e\in E}b_e^2$. Then, $U^*$ can be determined exactly by computing the fractional number that is closest to the midpoint, while having a denominator less than or equal to $\max_{e\in E}b_e$. After having $U^*$, one can simply apply edge splitting and spanning tree construction to $G(\{\lfloor U^*b_e\rfloor\}_{e\in E})$ to derive the pipeline schedule. The following theorem gives a bound on how close $\frac{M}{Nk}\cdot U^*$ is to bandwidth optimality (\ref{eq:lowerbound}):
\begin{restatable}{theorem}{thmfixedkapprox}\label{thm:fixedkapprox}
    Let $\frac{M}{Nk}\cdot U^*$ be the lowest bandwidth runtime that can be achieved with $k$ out-trees per $v\in V_c$. Then,
    \[
        \frac{M}{Nk}\cdot U^*\leq\frac{M}{N}\max_{S\subset V,S\not\supseteq V_c}\frac{|S\cap V_c|}{B^+_G(S)}+\frac{M}{Nk}\cdot\frac{1}{\min_{e\in E}b_e}.
    \]
\end{restatable}
\section{Runtime Analysis}

In this section, we give a runtime analysis of different parts of the algorithm. To summarize, all parts are strongly polynomial-time. Note that the runtime bounds discussed in this section could be loose in many respects. The analysis of this section serves to show that the runtime is polynomial in topology size. We leave tighter runtime bounds for future work.

\paragraph{Optimality Binary Search}The key part of optimality binary search is to compute $\min_{v\in V_c} F(s,v;\vec{G}_x)$, which involves computing maxflow from $s$ to every compute node in $V_c$. Assuming the use of preflow-push algorithm~\cite{preflowpush} to solve network flow, the time complexity to compute $\min_{v\in V_c} F(s,v;\vec{G}_x)$ is $\cO(N|V|^2|E|)$. Note that in practice, one can compute the maxflow from $s$ to each $v\in V_c$ in parallel to significantly speed up the computation. As for how many times $\min_{v\in V_c} F(s,v;\vec{G}_x)$ is computed, observe that the binary search terminates when range is smaller than $1/\min_{v\in V_c}B^-_G(v)^2$. The initial range of binary search is bounded by interval $(0,N)$, so the binary search takes at most $\lceil\log_2(N\min_{v\in V_c}B^-_G(v)^2)\rceil$ iterations. Because $\min_{v\in V_c}B^-_G(v)<|V|\max_{e\in E}b_e$ and $\cO(\log b_e)$ is trivial, the total runtime complexity is $\cO(N|V|^2|E|\log|V|)$.

\paragraph{Edge Splitting}In algorithm \ref{algo:removeswitch}, while we possibly add more edges to the topology, the number of edges is trivially bounded by $\cO(|V|^2)$. Thus, computing $M$ in theorem \ref{thm:multicapasplit} takes $\cO(N|V|^4)$, and $M$ is computed at most $\cO(|V_s||V|^4)$ times in the nested foreach loop. The total runtime complexity is $\cO(N|V_s||V|^8)$.
% In algorithm \ref{algo:removeswitch}, the inner foreach loop adds at most one edge per iteration. Thus, the total number of edges in $G^*$ is loosely bounded by $\cO(|V_s||E|^2)$. The key part again is to compute the $M$ in theorem \ref{thm:multicapasplit}. Because $G^*$ has $N$ vertices and $\cO(|V_s||E|^2)$ number of edges, the time complexity to compute $M$ is $\cO(N^3|V_s||E|^2)$. In algorithm \ref{algo:removeswitch}, $M$ is computed at most $|V_s||E|^2$ times. The total runtime complexity is $\cO(N^3|V_s|^2|E|^4)$.

\paragraph{Spanning Tree Construction}Upon completion of algorithm \ref{algo:removeswitch}, $G^*$ has $N$ vertices and hence $\cO(N^2)$ number of edges. In algorithm \ref{algo:construct}, $\mu$ only needs one maxflow to be computed. The runtime is thus $\cO(N^4)$. B{\'e}rczi \& Frank~\cite{berczi2010packing} proved that $\mu$ only needs to be computed $\cO(mn^2)$ times, where $m$ and $n$ are the number of edges and vertices respectively. Thus, the runtime complexity of algorithm \ref{algo:construct} is $\cO(N^8)$.

% \paragraph{Spanning Tree Construction}The runtime of this part depends on $k$, the number of spanning trees per compute node. For exact bandwidth optimality, $k$ reaches up to $\min_{v\in V_c}B^-_G(v)/\gcd(\{b_e\}_{e\in E})$ in worst-case scenario. One can also simply fix $k$ based on the method described in \S\ref{sec:fixedk}. Given $k$, we borrow the upper bound $\cO(N^2k^2|E|^2)$ from Tarjan~\cite{tarjan}.\todo{After edge splitting, how many edges are there?} However, our batched spanning tree construction is much faster in practice. A tighter upper bound or asymptotic bound is possible but not the focus of this paper.

\paragraph{Fixed-$k$ Optimality}The runtime of this part is similar to optimality binary search, with the exception that the binary search takes at most $\lceil\log_2(Nk\max_{e\in E}b_e^2)\rceil$ iterations instead. Since $\cO(\log b_e)$ and $\cO(\log k)$ are trivial, the total runtime complexity is $\cO(N|V|^2|E|\log N)$.

\bibliographystyle{acm}
\bibliography{ref}

\appendix

\section{Broadcast}\label{sec:broadcast}

Since allgather can be seen as simultaneous broadcast from every compute node, the techniques introduced in this paper can be easily applied to construct bandwidth optimal broadcast. Indeed, it has already been recognized in previous literature~\cite{bwoptBC} that the maximum set of edge-disjoint out-trees constitutes bandwidth optimal pipeline broadcast. In this section, we give a brief discussion on how to construct bandwidth optimal broadcast on switch topology. Similar to allgather lower bound (\ref{eq:lowerbound}), a lower bound for broadcast from a root $r$ is as followed:
\begin{equation}\label{eq:bclowerbound}
    T_B\geq M\brac{\min_{r\in S\not\supseteq V_c}B^+_G(S)}^{-1}=M\brac{\min_{S\cap V_c\notin\{\emptyset,V_c\}}B^+_G(S)}^{-1}.
\end{equation}
The equality is due to $G$ being Eulerian. By (\ref{eq:bclowerbound}), the bandwidth runtime of broadcast is also bounded by the bottleneck cut. The bottleneck cut can be easily found by computing $\min_{v\in V_c}F(r,v;G)$. Like in allgather case, we can iteratively split off incoming and outgoing edges of switch nodes while maintaining $\min_{v\in V_c}F(r,v;G)$ unchanged. After getting the compute-node-only topology $G^*$, an allgather pipeline schedule can be constructed simply by generating the maximum set of edge-disjoint out-trees rooted at $r$. Such a schedule reaches the lower bound (\ref{eq:bclowerbound}) and is thus bandwidth optimal. As for the algorithm to generate edge-disjoint out-trees with a single root, there are plenty of algorithms in existing literature~\cite{tarjan,TONG198373,Gabow}.
\section{Allreduce}\label{sec:allreduce}

The algorithm described in this paper can be used to construct an allreduce pipeline schedule. First of all, one can apply the method in appendix A of Zhao et al.~\cite{zhao2022optimal} to construct a reduce-scatter pipeline schedule. The idea is to construct an allgather pipeline schedule on transpose graph of the topology and then reverse all the communications. Such a reduce-scatter schedule is also bandwidth optimal. One can then construct an allreduce schedule by simply concatenating the reduce-scatter and allgather schedules (RS+AG). In this section, we focus on discussing the bandwidth optimality of allreduce and how close our method can reach. We start with the following theorem:
\begin{restatable}{theorem}{thmARlowerboundone}\label{thm:ARlowerboundone}
    Given Eulerian digraph $G=(V=V_s\cup V_c,E)$, the bandwidth runtime of allreduce satisfies
    \begin{equation}\label{eq:ARlowerboundone}
        T_B\geq M\brac{\min_{S\cap V_c\notin\{\emptyset,V_c\}} B^+_G(S)}^{-1}.
    \end{equation}
\end{restatable}
$\min_{S\cap V_c\notin\{\emptyset,V_c\}} B^+_G(S)$ can be easily calculated by computing the minimum maxflow from an arbitrary compute node to every other one in $V_c$, i.e. $\min_{v\in V_c}F(v_0,v;G)$ for arbitrary $v_0\in V_c$. There is another lower bound proven by Patarasuk \& Yuan~\cite{arlowerboundtwo}:
\begin{restatable}[Patarasuk \& Yuan~\cite{arlowerboundtwo}]{theorem}{thmARlowerboundtwo}\label{thm:ARlowerboundtwo}
    To perform allreduce on $N$ processes, there exists a process sending at least $2M(N-1)/N$ amount of data and a process receiving at least $2M(N-1)/N$ amount of data.
\end{restatable}
\begin{restatable}{corollary}{corARlowerboundtwo}\label{cor:ARlowerboundtwo}
    Given Eulerian digraph $G=(V=V_s\cup V_c,E)$, the bandwidth runtime of allreduce satisfies
    \begin{equation}\label{eq:ARlowerboundtwo}
        T_B\geq\frac{2M(N-1)}{N}\brac{\max_{v\in V_c}\min_{S\cap V_c=\{v\}}B^+_G(S)}^{-1}.
    \end{equation}
\end{restatable}
$\min_{S\cap V_c=\{v\}}B^+_G(S)$ can also be computed by a maxflow with $v$ as source and all other compute nodes connected to a sink with $\infty$ capacity. In practice, we found (\ref{eq:ARlowerboundone}) and (\ref{eq:ARlowerboundtwo}) alternately serve as the greater lower bound, depending on the network topology. Based on these lower bounds, we have the following theorem:
\begin{restatable}{theorem}{thmARoptimal}\label{thm:ARoptimal}
    Given Eulerian digraph $G=(V=V_s\cup V_c,E)$, let $S^*=\argmax_{S\subset V,S\not\supseteq V_c}|S\cap V_c|/B^+_G(S)$. Then, concatenating bandwidth optimal reduce-scatter and allgather gives bandwidth optimal allreduce if one of the following is true:
    \begin{enumerate}[label=\normalfont(\alph*)]
        \item\label{ARoptimala} $|S^*\cap V_c|=N/2$;
        \item\label{ARoptimalb} $S^*\cap V_c=\{v'\}$ and $\displaystyle\min_{S\cap V_c=\{v'\}}B^+_G(S)=\max_{v\in V_c}\min_{S\cap V_c=\{v\}}B^+_G(S)$ for some $v'\in V_c$.
    \end{enumerate}
\end{restatable}
Note that $|S^*\cap V_c|$ cannot be less than $N/2$; otherwise, $|S\cap V_c|/B^+_G(S)<|\bar{S}\cap V_c|/B^+_G(\bar{S})$. Condition (\ref{ARoptimalb}) is trivial if the topology is symmetric that $\min_{S\cap V_c=\{v\}}B^+_G(S)$ is the same for all $v\in V_c$. The switch topology in figure \ref{fig:topo} satisfies condition \ref{ARoptimala}. Thus, RS+AG gives bandwidth optimal allreduce in this topology.\\

Another way of constructing allreduce is concatenating reduce and broadcast (RE+BC) as in Blink~\cite{blink}. Here, we show that optimal RS+AG is always strictly better than optimal RE+BC. From \S\ref{sec:broadcast}, we know the optimal bandwidth runtime of broadcast is (\ref{eq:bclowerbound}). Observe that
\[
    M\brac{\min_{S\cap V_c\notin\{\emptyset,V_c\}}B^+_G(S)}^{-1}=\frac{M}{N}\max_{S\cap V_c\notin\{\emptyset,V_c\}}\frac{N}{B^+_G(S)}>\frac{M}{N}\max_{S\subset V,S\not\supseteq V_c}\frac{|S\cap V_c|}{B^+_G(S)}.
\]
Thus, (\ref{eq:bclowerbound}) is strictly greater than (\ref{eq:lowerbound}). Take the switch topology in figure \ref{fig:topo} as example. The optimal bandwidth runtime of broadcast is $M/4b$ by (\ref{eq:bclowerbound}), which is 2x than the allgather optimality $(M/N)(4/4b)$ by (\ref{eq:lowerbound}). \\

Allreduce optimality remains largely unsolved. We leave this topic for future work. In summary, we ask the following open questions:
\begin{enumerate}
    \item Does bandwidth optimal RS+AG equal bandwidth optimal allreduce?
    \item What is the necessary and sufficient condition for bandwidth optimal RS+AG to be a bandwidth optimal allreduce?
    \item What is the bandwidth optimality of allreduce in general?
    \item How to construct bandwidth optimal allreduce in general?
\end{enumerate}
\section{Proofs}

\thmbinarysearch*
\begin{proof}
    $\Rightarrow$: Suppose $1/x<\max_{S\subset V,S\not\supseteq V_c}|S\cap V_c|/B^+_G(S)$. Let $S'\subset V,S'\not\supseteq V_c$ be the set that $1/x<|S'\cap V_c|/B^+_G(S')$. Pick arbitrary $v'\in V_c-S'$. Consider the maxflow $F(s,v';\vec{G}_x)$ and $s$-$v'$ cut $(A,\bar{A})$ in $G$ that $A=S'+s$. We have
	\begin{equation}\label{eq:swppagarbitrarycut}
		c(A,\bar{A};\vec{G}_x)=c(S',\bar{A};\vec{G}_x)+\sum_{u\in\bar{A}\cap V_c}c(s,u;\vec{G}_x)=B^+_{G}(S')+|V_c-S'|x<|S'\cap V_c|x+|V_c-S'|x=|V_c|x.
	\end{equation}
	By min-cut theorem, $\min_{v\in V_c}F(s,v;\vec{G}_x)\leq F(s,v';\vec{G}_x)\leq c(A,\bar{A};\vec{G}_x)<|V_c|x$.\\

    $\Leftarrow$: Suppose $1/x\geq\max_{S\subset V,S\not\supseteq V_c}|S\cap V_c|/B^+_G(S)$. Pick arbitrary $v'\in V_c$. Let $(A,\bar{A})$ be arbitrary $s$-$v'$ cut and $S'=V\cap A=A-s$. It follows that $1/x\geq|S'\cap V_c|/B^+_G(S')$. Thus, following (\ref{eq:swppagarbitrarycut}),
	\[
		c(A,\bar{A};\vec{G}_x)=B^+_G(S')+|V_c-S'|x\geq |S'\cap V_c|x+|V_c-S'|x=|V_c|x.
	\]
	Because cut $(A,\bar{A})$ is arbitrary, we have $F(s,v';\vec{G}_x)\geq|V_c|x$. Because $v'$ is also arbitrary, we have $\min_{v\in V_c}F(s,v;\vec{G}_x)\geq|V_c|x$.
\end{proof}

\lmfraction*
\begin{proof}
    Because $a/b\neq c/d$, we have $ad-bc\neq 0$. Thus,
    \[
        \abs{\frac{a}{b}-\frac{c}{d}}=\abs{\frac{ad-bc}{bd}}\geq\frac{1}{bd}\geq\frac{1}{X^2}.
    \]
\end{proof}

\lmsmallestUk*
\begin{proof}
    Since $U/k=1/x^*$, we have $k=Ux^*$, so finding the smallest $k$ is to find the smallest $U$ such that (a) $Ux^*=Uq/p\in\N$ and (b) $Ub_e\in\N$ for all $e\in E$. Suppose $U=\alpha/\beta$ and $\gcd(\alpha,\beta)=1$. Because $\alpha,\beta$ are coprime, $Ub_e\in\N$ implies $\beta|b_e$ for all $e\in E$. Again, because $p,q$ are coprime, $Uq/p\in\N$ implies $p|\alpha$ and $\beta|q$. Thus, the smallest such $\alpha$ is $p$, and the largest such $\beta$ is $\gcd(q, \{b_e\}_{e\in E})$. The proposition follows.
\end{proof}

\thmbangjensentrees*
\thmexisttrees*
\begin{proof}
    $\Rightarrow$: Pick arbitrary $v\in V_c$. Given the family of edge-disjoint out-trees $\{T_{u,i}\}_{u\in V_c,i\in[k]}$, we push one unit of flow from $s$ to $v$ along the path from $u$ to $v$ within tree $T_{u,i}$ for each $u\in V_c,i\in[k]$. Thus, we have constructed a flow assignment with $|V_c|k$ amount of flow. Since $v\in V_c$ is arbitrary, we have $\min_{v\in V_c}F(s,v;\vec{D}_k)\geq|V_c|k$.\\

    $\Leftarrow$: Suppose $\min_{v\in V_c}F(s,v;\vec{D}_k)\geq|V_c|k$. It immediately implies that $\lambda(s,v;\vec{D}_k)\geq|V_c|k$ for all $v\in V_c$. Note that $T'=V_c$, so by theorem \ref{thm:bang-jensen}, a family $\cF$ of edge-disjoint out-trees rooted at $s$ exists that each $v\in V_c$ belongs to at least $|V_c|k$ of them. Since $d^+(s)=|V_c|k$ in $\vec{D}_k$, $\cF$ has exactly $|V_c|k$ edge-disjoint out-trees rooted at $s$ and each out-tree spans $V_c$. In addition, for each $v\in V_c$, since $c(s,v;\vec{D}_k)=k$, there are exactly $k$ out-trees in $\cF$ in which $v$ is the only child of root $s$. By removing the root $s$ from every out-tree in $\cF$, we have the family of edge-disjoint out-trees $\{T_{u,i}\}_{u\in V_c,i\in[k]}$ in $D$ as desired.
\end{proof}

\thmbangjensenedgesplit*
\thmedgesplit*
\begin{proof}
    Consider the flow network $\vec{D}_k$. We construct $\vec{D}_k'$ by adding a $k$-capacity edge from each $v\in V_c$ back to $s$. It is trivial to see that $\vec{D}_k'$ is Eulerian. By theorem \ref{thm:BJedgesplitting}, given $f=(w,t)$, there exists an edge $e=(u,w)$ such that $\lambda(s,v;\vec{D}'^{ef}_k)=\lambda(s,v;\vec{D}_k')$ for all $v\in V_c$. Observe that adding edges from $V_c$ to $s$ does not affect the edge-connectivity from $s$ to any $v\in V_c$, so for all $v\in V_c$,
    \[
        F(s,v;\vec{D}^{ef}_{k})=\lambda(s,v;\vec{D}^{ef}_k)=\lambda(s,v;\vec{D}'^{ef}_k)=\lambda(s,v;\vec{D}_k')=\lambda(s,v;\vec{D}_k)=F(s,v;\vec{D}_k).
    \]
    The theorem trivially follows.
\end{proof}

\thmmulticapasplit*
\begin{proof}
    First of all, one should note that for any $s$-$v$ cut $(A,\bar{A})$ with $v\in V_c$ and $A\subset V_s\cup V_c+s$, if $s,u,t\in A\land v,w\in\bar{A}$, then $(A,\bar{A})$ has the same capacity in $\vec{D}_k$ and $\widehat{D}_{(u,w),v}$ i.e. $c(A,\bar{A};\vec{D}_k)=c(A,\bar{A};\widehat{D}_{(u,w),v})$. Similarly, if $s,w\in A\land v,u,t\in\bar{A}$, then $c(A,\bar{A};\vec{D}_k)=c(A,\bar{A};\widehat{D}_{(w,t),v})$. \\
    
    $\geq$: Suppose we split off $(u,w),(w,t)$ by $M$ times and then $F(s,v';\vec{D}^{ef}_k)<|V_c|k$ for some $v'\in V_c$. Let $(A,\bar{A})$ be the min $s$-$v'$ cut in $\vec{D}^{ef}_k$ that $c(A,\bar{A};\vec{D}^{ef}_k)=F(s,v';\vec{D}^{ef}_k)<|V_c|k$. We assert that $(A,\bar{A})$ must cut through $(u,w)$ and $(w,t)$ such that either $s,u,t\in A\land v',w\in\bar{A}$ or $s,w\in A\land v',u,t\in\bar{A}$; otherwise, we have $F(s,v';\vec{D}_k)\leq c(A,\bar{A};\vec{D}_k)=c(A,\bar{A};\vec{D}^{ef}_k)<|V_c|k$ (note that splitting off $(u,w),(w,t)$ adds edge $(u,t)$). Suppose $s,u,t\in A\land v',w\in\bar{A}$, then $c(A,\bar{A};\vec{D}_k)=c(A,\bar{A};\widehat{D}_{(u,w),v'})$. It is trivial to see that $c(A,\bar{A};\vec{D}_k)=c(A,\bar{A};\vec{D}^{ef}_k)+M$. Thus, we have
    \[
        F(u,w;\widehat{D}_{(u,w),v'})\leq c(A,\bar{A};\widehat{D}_{(u,w),v'})=c(A,\bar{A};\vec{D}_k)=c(A,\bar{A};\vec{D}^{ef}_k)+M<|V_c|k+M,
    \]
    contradicting $M\leq\min_{v\in V_c} F(u,w;\widehat{D}_{(u,w),v})-|V_c|k$. For $s,w\in A\land v',u,t\in\bar{A}$, one can similarly show a contradiction by looking at $F(w,t;\widehat{D}_{(w,t),v'})$.\\
    
    $\leq$: Suppose we split off $(u,w),(w,t)$ by $M'>M$ times and the resulting graph is $D^{ef}$. It is trivial to see that $M'$ cannot be greater than $c(u,w;D)$ or $c(w,t;D)$. Suppose $M'>F(u,w;\widehat{D}_{(u,w),v'})-|V_c|k$ for some $v'\in V_c$. Consider the min $u$-$w$ cut $(A,\bar{A})$ with $c(A,\bar{A};\widehat{D}_{(u,w),v'})=F(u,w;\widehat{D}_{(u,w),v'})$. Because $(u,s),(u,t),(v',w)$ have $\infty$ capacity, we have $s,u,t\in A\land v',w\in\bar{A}$ and hence $c(A,\bar{A};\vec{D}_k)=c(A,\bar{A};\widehat{D}_{(u,w),v'})$. It is again trivial to see that $c(A,\bar{A};\vec{D}^{ef}_k)=c(A,\bar{A};\vec{D}_k)-M'$ and $(A,\bar{A})$ being an $s$-$v'$ cut in $\vec{D}^{ef}_k$. Hence,
    \[
        F(s,v';\vec{D}^{ef}_k)\leq c(A,\bar{A};\vec{D}^{ef}_k)=c(A,\bar{A};\vec{D}_k)-M'=c(A,\bar{A};\widehat{D}_{(u,w),v'})-M'<|V_c|k.
    \]
    One can show similar result for $M'>F(w,t;\widehat{D}_{(w,t),v'})-|V_c|k$.
\end{proof}

\thmtarjan*
\thminitialexist*
\begin{proof}
    $\Rightarrow$: Suppose $\min_{v\in V_c}F(s,v;\vec{D}_k)<|V_c|k$. Let $v'$ be the vertex that $F(s,v';\vec{D}_k)<|V_c|k$. By min-cut theorem, there exists an $s$-$v'$ cut $(A,\bar{A})$ in $\vec{D}_k$ such that $c(A,\bar{A};\vec{D}_k)=F(s,v';\vec{D}_k)<|V_c|k$. Let $S=V_c\cap A$, then $A=S+s$, $\bar{S}=V_c-S=V_c+s-A=\bar{A}$, and hence
    \[
        c(S,\bar{S};D)=c(A,\bar{A};\vec{D}_k)-\sum_{u\in\bar{A}}c(s,u;\vec{D}_k)<|V_c|k-|V_c-S|k=|S|k.
    \]
    $\Leftarrow$: Suppose there exists $S\subset V_c,S\neq V_c$ such that $c(S,\bar{S};D)<|S|k$. Pick arbitrary $v'\in V_c-S$. Consider $s$-$v'$ cut $(A,\bar{A})$ such that $A=S+s$. By min-cut theorem, we have
    \[
        F(s,v';\vec{D}_k)\leq c(A,\bar{A};\vec{D}_k)=c(S,\bar{S};D)+\sum_{u\in\bar{A}}c(s,u;\vec{D}_k)<|S|k+|V_c-S|k=|V_c|k.
    \]
\end{proof}

\thmpolynomialpacking*
\thmcomputemu*
\begin{proof}
    For simplicity of notation, let $L=\min\{c(S,\bar{S};D)-p(S;D):x\in S,y\in\bar{S},R_1\not\subseteq S\}$. We will prove (\ref{eq:computemu}) by showing that either $L=F(x,y;\overline{D})-\sum_{i\neq 1}m(R_i)$ or $L\geq F(x,y;\overline{D})-\sum_{i\neq 1}m(R_i)\geq m(R_1)$. Let $S\subset V_c$ be arbitrary that $x\in S,y\in\bar{S},R_1\not\subseteq S$, and Let $A=S\cup\{s_i\ |\ R_i\subseteq S\}$. It follows that $(A,\bar{A})$ is an $x$-$y$ cut in $\overline{D}$ and hence
    \begin{equation*}
    \begin{aligned}
        c(S,\bar{S};D)-p(S;D)
        &\textstyle=c(S,\bar{S};D)-\sum\{m(R_i)\ |\ R_i\subseteq S\}\\
        &\textstyle=c(S,\bar{S};D)+\sum\{m(R_i)\ |\ i\neq 1,R_i\not\subseteq S\}-\sum_{i\neq 1}m(R_i)\\
        &\textstyle=c(A,\bar{A};\overline{D})-\sum_{i\neq 1}m(R_i)\\
        &\textstyle\geq F(x,y;\overline{D})-\sum_{i\neq 1}m(R_i).
    \end{aligned}
    \end{equation*}
    The second equality is due to $R_1\not\subseteq S$, so $\sum\{m(R_i)\ |\ R_i\subseteq S\}=\sum\{m(R_i)\ |\ i\neq 1,R_i\subseteq S\}$. Since $S$ is arbitrary, we have $L\geq F(x,y;\overline{D})-\sum_{i\neq 1}m(R_i)$.\\
    
    Let $(A',\overline{A'})$ be the min $x$-$y$ cut in $\overline{D}$ and $S'=A'\cap V_c$. We assert that for any $i\neq 1,R_i\subseteq S'$, we have $s_i\in A'$; otherwise, by moving $s_i$ from $\overline{A'}$ to $A'$, we create a cut with lower capacity, contradicting $(A',\overline{A'})$ being min-cut. We also assert that for any $i\neq 1,R_i\not\subseteq S'$, we have $s_i\in\overline{A'}$; otherwise, there exists $v\in R_i-S'$ that $\infty$ edge $(s_i,v)$ crosses $(A',\overline{A'})$. Thus, we have
    \begin{equation}\label{eq:computemuproof2}
    \begin{aligned}
        \textstyle F(x,y;\overline{D})-\sum_{i\neq 1}m(R_i)
        &\textstyle=c(A',\overline{A'};\overline{D})-\sum_{i\neq 1}m(R_i)\\
        &\textstyle=c(S',\overline{S'};D)+\sum\{m(R_i)\ |\ i\neq 1,R_i\not\subseteq S'\}-\sum_{i\neq 1}m(R_i)\\
        &\textstyle=c(S',\overline{S'};D)-\sum\{m(R_i)\ |\ i\neq 1,R_i\subseteq S'\}.
    \end{aligned}
    \end{equation}
    Now, we consider two cases:
    \begin{enumerate}[label=(\alph*)]
        \item Suppose $R_1\not\subseteq S'$. Then, $c(S',\overline{S'};D)-p(S';D)\geq L$. By (\ref{eq:computemuproof2}), we have
        \[
            \textstyle L\geq F(x,y;\overline{D})-\sum_{i\neq 1}m(R_i)=c(S',\overline{S'};D)-\sum\{m(R_i)\ |\ i\neq 1,R_i\subseteq S'\}=c(S',\overline{S'};D)-p(S';D)
        \]
        Thus, $L=c(S',\overline{S'};D)-p(S';D)=F(x,y;\overline{D})-\sum_{i\neq 1}m(R_i)$ and (\ref{eq:computemu}) holds.
        \item Suppose $R_1\subseteq S'$. Because the existence of spanning trees is guaranteed, we have
        \[
            \textstyle c(S',\overline{S'};D)\geq p(S';D)=m(R_1)+\sum\{m(R_i)\ |\ i\neq 1,R_i\subseteq S'\}.
        \]
        Hence,
        \[
            \textstyle L\geq F(x,y;\overline{D})-\sum_{i\neq 1}m(R_i)=c(S',\overline{S'};D)-\sum\{m(R_i)\ |\ i\neq 1,R_i\subseteq S'\}\geq m(R_1).
        \]
        Thus, $\mu=\min\{g(x,y),m(R_1)\}$ and (\ref{eq:computemu}) also holds.
    \end{enumerate}
\end{proof}

\thmfixedk*
\begin{proof}
	$\Leftarrow$: Suppose $\{T_{u,i}\}_{u\in V_c,i\in[k]}$ is edge disjoint in $G(\{\lfloor Ub_e\rfloor\}_{e\in E})$, then
	\[
		T_B=\frac{M}{Nk}\cdot\max_{e\in E}\frac{1}{b_e}\sum_{T\in \{T_{u,i}\}}\I[e\in T]\leq\frac{M}{Nk}\cdot\max_{e\in E}\frac{\lfloor Ub_e\rfloor}{b_e}\leq\frac{M}{Nk}\cdot U.
	\]
	$\Rightarrow$: Suppose $\{T_{u,i}\}_{u\in V_c,i\in[k]}$ achieves $\frac{M}{Nk}\cdot U$ bandwidth runtime, then
	\[
		\max_{e\in E}\frac{1}{b_e}\sum_{T\in \{T_{u,i}\}}\I[e\in T]\leq U\qquad\Longrightarrow\qquad\sum_{T\in \{T_{u,i}\}}\I[e\in T]\leq Ub_e\quad\text{for all $e\in E$.}
	\]
	Since $\sum_{T\in \{T_{u,i}\}}\I[e\in T]$ must be an integer, the edge-disjointness trivially follows.
\end{proof}

\thmfixedkbinarysearch*
\begin{proof}
	$\Rightarrow$: The existence of edge-disjoint $\{T_{u,i}\}_{u\in V_c,i\in[k]}$ in $G(\{\lfloor Ub_e\rfloor\}_{e\in E})$ with $U<U^*$ simply contradicts $\frac{M}{Nk}\cdot U^*$ being the lowest bandwidth runtime. $\Leftarrow$: Let $\{T^*_{u,i}\}_{u\in V_c,i\in[k]}$ be the family of out-trees with lowest bandwidth runtime, then by theorem \ref{thm:fixedk}, it is edge-disjoint in $G(\{\lfloor Ub_e\rfloor\}_{e\in E})$ for all $U\geq U^*$. 
\end{proof}

\thmfixedkapprox*
\begin{proof}
    Let $U=\max_{e\in E}\lceil kb_e/x^*\rceil/b_e$ where $1/x^*=\max_{S\subset V,S\not\supseteq V_c}|S\cap V_c|/B^+_G(S)$. For each edge $(u,v)$ in $G(\lfloor Ub_e\rfloor)$, we have
    \[
        c(u,v;G(\lfloor Ub_e\rfloor))=\left\lfloor b_{(u,v)}\cdot\max_{e\in E}\frac{\lceil kb_e/x^*\rceil}{b_e}\right\rfloor\geq\left\lfloor b_{(u,v)}\cdot\frac{\lceil kb_{(u,v)}/x^*\rceil}{b_{(u,v)}}\right\rfloor=\lceil kb_{(u,v)}/x^*\rceil.
    \]
    Thus, each edge in $\vec{G}_k(\lfloor Ub_e\rfloor)$ has at least $k/x^*$ times the capacity in $\vec{G}_{x^*}$, so $\min_{v\in V_c}F(s,v;\vec{G}_k(\lfloor Ub_e\rfloor))\geq(k/x^*)\min_{v\in V_c}F(s,v;\vec{G}_{x^*})\geq|V_c|k$. Therefore, $\frac{M}{Nk}\cdot U$ is achievable and hence $U^*\leq U$ by theorem \ref{thm:fixedkbinarysearch}.
    \[
        \frac{U^*}{k}\bigg/\frac{1}{x^*}\leq\frac{U}{k}\bigg/\frac{1}{x^*}=\frac{\max_{e\in E}\lceil kb_e/x^*\rceil/b_e}{k/x^*}\leq\max_{e\in E}\frac{\lceil kb_e/x^*\rceil}{kb_e/x^*}\leq 1+\max_{e\in E}\frac{1}{kb_e/x^*}=1+\frac{x^*}{k\cdot\min_{e\in E}b_e}.
    \]
    The theorem trivially follows.
\end{proof}

\thmARlowerboundone*
\begin{proof}
    Given any $S\subset V$ satisfying $S\cap V_c\notin\{\emptyset,V_c\}$, consider the cut $(S,\bar{S})$. For any indivisible unit chunk $C$ of the data, there are three cases:
    \begin{enumerate}[label=(\alph*)]
        \item The allreduce result of $C$ is computed in $S$. In such case, at least one chunk $C$ must be sent from $\bar{S}$ to $S$ to be allreduced, and the allreduced chunk $C$ must also be sent from $S$ to $\bar{S}$.
        \item The allreduce result of $C$ is computed in $\bar{S}$. Similarly, at least one chunk $C$ needs to be sent from $S$ to $\bar{S}$ and then sent from $\bar{S}$ to $S$.
        \item The allreduce result of $C$ is computed in both $S$ and $\bar{S}$. Then, $S$ and $\bar{S}$ must also send at least one chunk $C$ to each other to be allreduced.
    \end{enumerate}
    Thus, every piece of data to be allreduced must pass through $(S,\bar{S})$ once in each direction. The lower bound (\ref{eq:ARlowerboundone}) immediately follows.
\end{proof}

\thmARlowerboundtwo*
\corARlowerboundtwo*

\thmARoptimal*
\begin{proof}
    If $|S^*\cap V_c|=N/2$, then the reduce-scatter/allgather bandwidth optimal runtime is $\frac{M/2}{B^+_G(S^*)}$. RS+AG gives allreduce runtime
    \[
        T_B=2\cdot\frac{M/2}{B^+_G(S^*)}=\frac{M}{B^+_G(S^*)}\leq M\brac{\min_{S\cap V_c\notin\{\emptyset,V_c\}} B^+_G(S)}^{-1}.
    \]
    By lower bound (\ref{eq:ARlowerboundone}), the allreduce is bandwidth optimal.\\

    If $S^*\cap V_c=\{v'\}$ and $\displaystyle\min_{S\cap V_c=\{v'\}}B^+_G(S)=\max_{v\in V_c}\min_{S\cap V_c=\{v\}}B^+_G(S)$, similarly, RS+AG gives allreduce runtime
    \[
        T_B=2\cdot\frac{M}{N}\cdot\frac{N-1}{B^+_G(S^*)}\leq\frac{2M(N-1)}{N}\brac{\min_{S\cap V_c=\{v'\}}B^+_G(S)}^{-1}=\frac{2M(N-1)}{N}\brac{\max_{v\in V_c}\min_{S\cap V_c=\{v\}}B^+_G(S)}^{-1}.
    \]
    By lower bound (\ref{eq:ARlowerboundtwo}), the allreduce is bandwidth optimal.
\end{proof}

\end{document}